\documentclass{lmcs}
\pdfoutput=1

\usepackage{lastpage}
\lmcsdoi{19}{4}{4}
\lmcsheading{}{\pageref{LastPage}}{}{}%
{Nov.~08,~2021}{Oct.~18,~2023}{}

\usepackage[utf8]{inputenc}
\usepackage[hidelinks,breaklinks]{hyperref}
\usepackage{breakurl}
\usepackage{microtype}
\title{Locality and Centrality: The Variety \textbf{ZG}}
\author[A.~Amarilli]{Antoine Amarilli\lmcsorcid{0000-0002-7977-4441}}[a]
\address{\lsuper{a} LTCI, Télécom Paris, Institut polytechnique de Paris, France}
\email{antoine.amarilli@telecom-paris.fr}
\author[C.~Paperman]{Charles Paperman\lmcsorcid{0000-0002-6658-5238}}[b]
\address{\lsuper{b} Univ. Lille, CNRS, INRIA, Centrale Lille, UMR 9189 CRIStAL, F-59000 Lille, France}
\email{charles.paperman@univ-lille.fr}
\thanks{We thank Jean-Éric Pin and Jorge Almeida for their fruitful
advice, and thank the anonymous reviewers of ICALP and LMCS for their very helpful feedback.}

\newcommand{\Li}{\mathbf{LI}}
\newcommand{\MNil}{\mathbf{MNil}}
\newcommand{\DA}{\mathbf{DA}}
\newcommand{\D}{\mathbf{D}}
\newcommand{\J}{\mathbf{J}}
\newcommand{\Com}{\mathbf{Com}}

\newcommand{\card}[1]{\left|#1\right|}
\newcommand{\join}{\lor}
\newcommand{\NN}{\mathbb{N}}
\newcommand{\CC}{\mathrm{C}}

\newcommand{\W}{\mathbf{W}}

\newcommand{\ZG}{\mathbf{ZG}}
\newcommand{\ZGp}{\mathbf{ZG}_p}
\newcommand{\ZE}{\mathbf{ZE}}
\newcommand{\LZG}{\mathbf{LZG}}
\newcommand{\LZGp}{\mathbf{LZG}_p}

\newcommand{\ZGD}{\Dop{\ZG}}
\newcommand{\LMNil}{\mathbf{LMNil}}
\newcommand{\MNilD}{\Dop{\MNil}}
\newcommand{\ZGpD}{\Dop{\ZG_p}}

\newcommand{\Dop}[1]{#1 *\D}
\newcommand{\ZZ}{\mathbb{Z}}
\newcommand{\EE}{\mathrm{E}}
\newcommand{\V}{\mathbf{V}}

\newcommand\restr[2]{{%
  \kern-\nulldelimiterspace %
  #1_{|#2} %
  }}
\newcommand{\act}{\textrm{act}}

\makeatletter
\let\enddoc@text\relax
\makeatother

\usepackage{amssymb}
\usepackage{colonequals}
\usepackage[bibliography=common]{apxproof}

\makeatletter
  \def\axp@newtheoremrep#1[#2]#3[#4]{%
    \axp@newtheoremrep@definetheorem{#1}{#2}{#3}{#4}%
    \expandafter\pretocmd\csname #1\endcsname{\noproofinappendix}{}{}%
    \axp@newtheorem*{axp@#1rp}{#3}%
    \axp@forward@setup{#1}{#2}{#3}{#4}%
    \NewEnviron{#1rep}[1][]{%
      \ifthenelse{\equal{\axp@repeqn}{same}}{%
        \setcounter{axp@equation}{\value{equation}}%
      }{}%
      \addtocounter{axp@rpcounter}{1}%
      \ifx\relax##1\relax
        \axp@with@forward{#1}{\begin{#1}}\label{axp@r\roman{axp@rpcounter}}%
      \else
        \axp@with@forward{#1}{\begin{#1}[{##1}]}\label{axp@r\roman{axp@rpcounter}}%
      \fi
      \mainbodyrepeatedtheorem
      \BODY\end{#1}%
      \global\toggletrue{axp@seenreptheorem}%
      \global\expandafter\let\csname rplet\roman{axp@rpcounter}%
                             \endcsname
      \BODY
      \axp@writesection%
      \immediate\write\axp@proofsfile{\noexpand\makeatletter}%
      \ifthenelse{\equal{\axp@repeqn}{same}}{%
        \immediate\write\axp@proofsfile{%
          \noexpand\setcounter{axp@equationx}{\value{equation}}%
          \noexpand\setcounter{equation}{\theaxp@equation}%
        }%
      }{}%
      \ifbool{axp@forward@suppress}{%
        \global\def\axp@refstar{\ref*}
      }{%
        \global\def\axp@refstar{\bf\ref}
      }
      \immediate\write\axp@proofsfile{{%
        \ifdefined\theopargself
          \noexpand\theopargself
        \else
          \noexpand\pretocmd{\noexpand\@begintheorem}{%
            \noexpand\patchcmd{\noexpand\thmhead}{\noexpand\@acmplainnotefont}{}{}{}%
            \noexpand\patchcmd{\noexpand\thmhead}{\noexpand\the\noexpand\thm@notefont}{}{}{}%
            \noexpand\patchcmd{\noexpand\thmhead}{(}{}{}{}%
            \noexpand\patchcmd{\noexpand\thmhead}{)}{}{}{}%
          }{}{}
        \fi
        \noexpand\begin{axp@#1rp}
          [%
            {\!\!\noexpand\axp@refstar{axp@r\roman{axp@rpcounter}}}%
            \@ifnotempty{##1}{%
              \ifdefined\theopargself
              \else
                \ifdefined\@acmplainnotefont
                  \noexpand\@acmplainnotefont
                \else
                  \noexpand\ifdefined\noexpand\thm@notefont
                    \noexpand\the\noexpand\thm@notefont
                  \noexpand\fi
                \fi
              \fi
              {} \unexpanded{{##1}}%
            }%
          ]%
          \noexpand\axp@forward@target{axp@fw@r\roman{axp@rpcounter}}{}%
          \noexpand\axp@redefinelabels
          \expandafter\noexpand\csname rplet\roman{axp@rpcounter}%
                               \endcsname
        \noexpand\end{axp@#1rp}
      }}%
      \ifthenelse{\equal{\axp@repeqn}{same}}{%
        \immediate\write\axp@proofsfile{%
          \noexpand\setcounter{equation}{\value{axp@equationx}}%
        }%
      }{}%
    }%
  }
\makeatother

\theoremstyle{thmC}
\newtheoremrep{lemma}[thm]{Lemma}

\begin{document}
\maketitle

\begin{abstract}
  We study the variety $\ZG$ of monoids where the elements that belong to a
  group are \emph{central}, i.e., commute with all other elements. We show that~$\ZG$
  is \emph{local}, that is, the semidirect product $\ZGD$ of $\ZG$  by definite
  semigroups is equal to $\LZG$, the variety of semigroups where all local
  monoids are in~$\ZG$. Our main result is thus: $\ZGD = \LZG$. We prove this
  result using Straubing's delay theorem, by considering paths in the category
  of idempotents. In the process, we obtain the characterization $\ZG = \MNil
  \lor \Com$, and also characterize the $\ZG$
  languages, i.e., the languages whose syntactic monoid is in~$\ZG$: they are
  precisely the languages that are finite unions of disjoint shuffles of
  singleton languages and regular commutative languages.
\end{abstract}

\section{Introduction}
In this paper, we study a variety of monoids called $\ZG$.
It is defined by requiring that the elements of the monoid that belong
to a group are \emph{central}, i.e., commute with all other elements of the
monoid. The notation $\ZG$ thus stands for \emph{Zentral Group},
inspired by the classical notion of centrality in group theory.
We can also define $\ZG$ with the equation $x^{\omega+1} y = y
x^{\omega+1}$ on all elements~$x$ and~$y$, where $\omega$ is the idempotent power of the monoid.

The variety $\ZG$ has been introduced by Auinger~\cite{auinger2000semigroups} as a subvariety of interest
of the broader and better-known class $\ZE$ of semigroups where the idempotent elements are
central.
The study of $\ZE$ was initiated by Almeida and
Azevedo~\cite{almeida1987implicit}.
Straubing studied in particular the variety $\MNil$ (called simply $\V$
in~\cite{Straubing82})
of regular languages generated by finite languages,
and showed that it is exactly the variety of aperiodic monoids in $\ZE$.
From this, a systematic investigation of the subclasses of $\ZE$
was started by Almeida and pursued by Auinger: see~\cite[page 211]{Almeida96}
and~\cite{auinger2000semigroups, auinger2002}.

Our specific motivation to investigate $\ZG$ is our recently published
study~\cite{incrementalarxiv}
of the \emph{dynamic membership problem for regular languages}.
In this problem, introduced in~\cite{skovbjerg1997dynamic},
we study how to apply substitution update operations on an input word
while maintaining the information of whether it belongs to a fixed regular
language. In~\cite{incrementalarxiv}, we show that this can be performed in
constant time per update for monoids in~$\ZG$, and extend this to 
semigroups and languages. Further, $\ZG$ turns out to be a plausible
boundary to characterize constant-time complexity for monoids, semigroups, and
languages. However, the case of semigroups requires a study of the
so-called \emph{semidirect product} of~$\ZG$ by definite ($\D$) semigroups, which
we denote by $\ZGD$. Thus, some of our results on semigroups and languages
in~\cite{incrementalarxiv} require an deeper understanding of $\ZG$ and $\ZGD$,
which is the focus of this work.

The semidirect product operation on varieties used to define~$\ZGD$ intuitively corresponds to composing finite
automata via a kind of cascade operation. Its study is the subject of a large
portion of semigroup theory, inspired by the classical study
of semidirect products in group theory.
There are also known results to understand 
the semidirect product specifically \emph{by} $\D$.
For instance, the
\emph{Derived category theorem}~\cite{Tilson}
studies it as a decisive step towards proving the decidability of membership
to an arbitrary semidirect product, i.e., deciding if a given monoid belongs to the product.
The product by~$\D$ also arises naturally in several other
contexts:
the dotdepth hierarchies~\cite{straubing1985finite},
the circuit complexity
of regular languages~\cite{Straubing94},
or the study of the successor relations in first-order logic~\cite{TherienW98, TherienW01, KufleitnerL13}.

Understanding this product with~$\D$ is notoriously complicated. For instance,
it requires specific dedicated work for some varieties like $\J$ or $\Com$~\cite{Knast83a, Knast83b, therien85}.
Also, this product does not preserve the decidability of membership,
i.e.,
Auinger~\cite{Auinger10} proved that
there are varieties $\V$ such that membership in~$\V$ is decidable, but
the analogous problem for~$\Dop{\V}$ is undecidable.
For the specific case of the varieties $\ZG$, $\ZE$, or even $\MNil$,
we are not aware of prior results describing their semidirect product with~$\D$.

\paragraph*{Locality.}
Existing work has nevertheless identified some cases where the $*\,\D$ operator can be
simplified to a much nicer \emph{local operator}, that preserves the decidability of membership and
is easier to understand.
For any semigroup $S$, the \emph{local monoids} of $S$ are the subsemigroups
of~$S$ of the form $eSe$ with $e$ an idempotent element of~$S$.
For a variety~$\V$,
we say that a semigroup belongs to $\mathbf{LV}$ if all its local monoids are in $\V$.
It is not hard to notice that the variety  $\Dop{\V}$ is always a subvariety of
$\mathbf{LV}$, i.e., that every monoid in~$\Dop{\V}$ must also be in $\mathbf{LV}$.
In some cases,
we can show a \emph{locality result} stating that the other
direction also holds, so that $\Dop{\V} = \mathbf{LV}$. In those cases we say that
the variety $\V$ is \emph{local}.
The locality of the variety of monoids $\DA$~\cite{Almeida96}
is a famous result that has deep implications
in logic and complexity~\cite{TherienW98, Dartois15, GrosshansMS17}
and has inspired recent follow-up work~\cite{PlaceS16}.
Locality results are also known for other varieties, for instance
the variety of \emph{semi-lattice monoids} (monoids that are both idempotent and
commutative)~\cite{McNaughton74, Brzozowski73}, any \emph{sub-varieties
of groups}~\cite[Theorem 10.2]{straubing1985finite}, or the
\emph{$\mathcal{R}$-trivial variety}~\cite{Stiffler73, Steinberg04, Straubing15}.
Thus, a hope to understand the variety $\ZGD$ is to establish a locality result
of this type for~$\ZG$.

\paragraph*{Contributions.}
This paper shows the locality of the variety~$\ZG$. 
This is achieved by showing a slightly stronger statement: each variety $\ZG_p$
is local, where $\ZG_p$ is the variety of $\ZG$ monoids where each subgroup has
a period dividing~$p$:

\begin{thm}
    \label{thm:locality2}
    For every $p>0$, we have $\LZG_p = \ZGpD$.
\end{thm}

As $\ZG$ is the union of the $\ZG_p$, this easily implies:

\begin{cor}
    \label{cor:locality}
    We have $\LZG = \ZGD$.
\end{cor}
\pagebreak
Further, the variety $\MNil$ of aperiodic monoids in~$\ZE$ 
(introduced in~\cite{Straubing82} and 
mentioned earlier) is exactly $\ZG_1$. Thus:

\begin{cor}
    \label{cor:mnil}
    We have $\LMNil = \MNilD$.
\end{cor}

In the process of proving Theorem~\ref{thm:locality2}, we obtain a characterization of $\ZG_p$-congruences,
i.e., congruences $\sim$ on~$\Sigma^*$ where the quotient
$\Sigma^*/\sim$ is a monoid of~$\ZG_p$. We show that they are always refined by a
so-called \emph{$n,p$-congruence}, which identifies the number of occurrences of
the \emph{frequent letters} (the ones occurring $> n$ times in the word)
modulo~$p$, and also identifies the exact subword formed by the \emph{rare letters} (the ones
occurring $\leq n$ times). Thanks to this 
(Theorem~\ref{thm:congruence}), we also obtain a characterization of the
languages of~$\ZG$, i.e., the languages whose syntactic monoid is in~$\ZG$: they
are exactly the finite unions of disjoint shuffles of singleton languages and
commutative languages (Corollary~\ref{cor:zgcarac}). We also characterize
$\ZG$ as a variety of monoids
(Corollary~\ref{cor:zgmoncarac}): $\ZG = \MNil \lor \Com$, for
$\MNil$ defined in~\cite{Straubing82} and $\Com$ the variety of commutative
monoids.

\paragraph*{Paper structure.}
We give preliminaries in Section~\ref{sec:prelim} and formally define
the variety~$\ZG$. We then give in Section~\ref{sec:carac} our characterizations
of~$\ZG$ 
via $n,p$-congruences. We then define in
Section~\ref{sec:zgli} the varieties $\ZGD$, $\LZG$, $\ZGpD$, and $\LZGp$
used in Theorem~\ref{thm:locality2} and Corollary~\ref{cor:locality}, which we prove in the
rest of the paper. We first introduce the framework of Straubing's
delay theorem used for our proof in Section~\ref{sec:straubing}, and
rephrase our result as a claim
(Claim~\ref{clm:hard})
on paths in the category of idempotents. We then
study in Section~\ref{sec:choosen} how to pick a sufficiently large value of~$n$
as a choice of our $n,p$-congruence, 
show in Section~\ref{sec:auxpaths} two lemmas on paths in the 
category of idempotents, and finish the proof in Section~\ref{sec:actualproof}
by two nested inductions. We
conclude in Section~\ref{sec:conclusion}.

\section{Preliminaries}
\label{sec:prelim}

For a complete presentation of the basic concepts (automata, monoids,
semigroups, groups, etc.) the reader can refer to
the book of Pin~\cite{Pin84} or to the more recent lecture notes~\cite{mpri}.
Except for the free monoid,
all semigroups, groups, and monoids that we consider are finite.

\paragraph*{Semigroups, monoids, varieties.}
A \emph{semigroup} $S$ is a set equipped with an internal associative law
(written multiplicatively). A
\emph{monoid} $M$ is a semigroup with an \emph{identity element} $1$, i.e., 
an element with $1 x = x 1 = x$ for all $x \in M$: note that the identity
element is necessarily unique.
A \emph{variety of semigroups} (resp., \emph{variety of monoids}) is a class of
semigroups (resp., monoids)
 closed under direct product, quotient, and subsemigroup (resp., submonoid).

For a semigroup $S$, we call $x \in S$ \emph{idempotent} if $x^2 = x$.
We call the \emph{idempotent power} of $x \in S$ the unique idempotent element which is a power of~$x$. (This means that $x$ is idempotent iff it is its own idempotent power.)
Now, the \emph{idempotent power of~$S$} is the least integer $\omega > 0$ 
such that for any element $x \in S$,
the element $x^\omega$ is an idempotent power of~$x$.
We write $x^{\omega+k}$ for any $k \in \ZZ$
to mean $x^{\omega+k'}$ where $k' \in \{0, \ldots, \omega-1\}$ is the remainder of $k$
in the integer division by $\omega$.
For example, $x^{\omega-1}$ simply denotes $x^{2\omega-1}$, where we have
$2\omega-1 > 0$; in particular, as expected, we have $x^{\omega-1} x =
x^\omega$.

\paragraph*{Languages and congruences.}
We denote by $\Sigma$ an alphabet and by $\Sigma^*$ the set of all finite words on~$\Sigma$.
A \emph{factor} of a word is a contiguous subword of that word.
For $w \in \Sigma^*$, we denote by $\card{w}$ the length of~$w$.
For $u, v \in \Sigma^*$, we say that $u$ is a \emph{subword} of~$v$ if, letting
$n \colonequals \card{u}$, there are 
$1 \leq i_1 < \cdots < i_n \leq \card{v}$ such that $u = v_{i_1} \cdots v_{i_n}$.
For $w \in \Sigma^*$ and $a \in \Sigma$, we denote by $\card{w}_a$ the number of occurrences of~$a$ in~$w$.
A \emph{language} $L$ is a subset of~$\Sigma^*$.

A \emph{congruence} on a finite alphabet~$\Sigma$ is an equivalence relation
on~$\Sigma^*$ which satisfies \emph{compositionality}, i.e., it is compatible
with the concatenation of words in the following sense: for any words $x$, $y$,
$z$, and $t$ of~$\Sigma^*$, if $x \sim y$ and $z \sim t$, then $xz \sim yt$.
It is a \emph{finite index congruence} if it has a finite number of equivalence classes.
For a given finite index congruence~$\sim$, the quotient $\Sigma^*/{\sim}$ is a
finite monoid, whose law corresponds to concatenation over~$\Sigma^*$ and whose
identity element is the class of the empty word.

The \emph{syntactic congruence} of a regular language $L$ over an alphabet $\Sigma$ is
the congruence on $\Sigma$ where $u, v \in \Sigma^*$ are equivalent if for all
$s, t \in \Sigma^*$, we have $sut \in L$ iff $svt \in L$. As $L$ is regular,
this congruence is a finite index congruence.
The \emph{syntactic monoid} of~$L$ over~$\Sigma^*$ is then
the quotient of~$\Sigma^*$ by the
syntactic congruence for~$L$.
The \emph{syntactic semigroup} of~$L$ is the quotient of $\Sigma^+$, the
non-empty words over~$\Sigma$, by the syntactic congruence.

A \emph{variety of regular languages} is a class of regular languages which is
closed under Boolean operations, left and right derivatives (also called left
and right quotients), and inverse homomorphisms.
By Eilenberg's theorem \cite{Eilenberg76}, a variety $\V$ of monoids 
defines a \emph{variety of regular languages}, namely, the languages whose syntactic
monoid is in $\V$.
Following standard practice, we abuse
notation and identify varieties of monoids with varieties of languages following
this correspondence, i.e., we write $\V$ both for the variety of monoids and
the variety of languages.

Letting $\V$ be a variety of monoids,
we say that a finite index congruence $\sim$ on~$\Sigma$ is a
\emph{$\V$-congruence} if the quotient
$\Sigma^*/{\sim}$ is a monoid in $\V$.
For a given $\V$-congruence $\sim$, the map $\eta:\Sigma^*\to \Sigma^*/{\sim}$,
defined by associating each word with its equivalence class, is a surjective
morphism to a monoid of $\V$.
Hence, each equivalence class is a language of $\V$, since it is recognized by
$\Sigma^*/{\sim}$.

\paragraph*{The variety $\ZG$.}
In this paper, we study the variety of monoids $\ZG$
defined (via Reiterman's theorem~\cite{reiterman1982birkhoff})
by the equation: $x^{\omega+1} y = y x^{\omega+1}$.
Intuitively, this says that the elements
of the form $x^{\omega+1}$ are central, i.e., commute with all other elements.
This clearly implies the same for elements of
the form $x^{\omega+k}$ for any $k\in\ZZ$, as we will
implicitly use throughout the paper:

\begin{clm}
\label{clm:eqnk}
	For any monoid $M$ in $\ZG$, for $x, y \in M$, and $k \in \ZZ$, we have:
 $x^{\omega+k} y = y x^{\omega+k}$.
\end{clm}

\begin{proof}
  Recall that $x^{\omega+k}$ denotes $x^{\omega+k'}$ where $k' \in \{0, \ldots,
  \omega-1\}$ is the remainder of $k$
in the integer division by $\omega$.
Now, we can write $x^{\omega+k'}$ as $(x^{\omega+k'})^{\omega+1}$, because the latter is equal to
$x^{(\omega+k')\times (\omega+1)}$ which is indeed equal to $x^{\omega+k'}$.
  Thus, by setting $x' \colonequals (x^{\omega+k'})$ and applying the equation
  of~$\ZG$, we conclude.
\end{proof}

Note that the elements of the form $x^{\omega+1}$
are precisely the \emph{group elements}, namely, the elements of the monoid that are within a (possibly trivial)
group that is a subsemigroup of the monoid. In particular, the neutral
elements of these groups are not necessarily the neutral element of the
monoid.

Claim~\ref{clm:eqnk}
motivates the name $\ZG$, which stands for ``Zentral Group'': it
follows the
traditional notation $Z(\cdot)$ for central subgroups, and extends the variety $\ZE$
introduced in~\cite[p211]{almeida1994finite} which only requires idempotents to
be central.
Thus, we have $\ZG \subsetneq \ZE$, and
non-commutative groups are examples of monoids that are in $\ZE$ but not in~$\ZG$.

In addition to the class $\ZG$, we will specifically study
the subclasses $\ZG_p$ for $p > 0$ defined by imposing the equation
$x^{\omega+p} = x^\omega$ in addition to the $\ZG$ equation. Intuitively,
$\ZG_p$ is the variety of monoids (and associated regular languages) where group
elements commute and where the \emph{period} of all group elements divides $p$,
where the \emph{period} of a group element~$x$ is the smallest integer $p'$ such that
$x^{p'} = x$. 
Clearly, $\ZG = \bigcup_{p > 0} \ZG_p$. Further, any finite language is in
$\ZG_p$ for any $p>0$.

The \emph{period} of a semigroup is the least common multiple of all periods
of group elements. Note that the period of a semigroup always divides~$\omega$,
because the period of any group element is equal to its idempotent power, hence
divides~$\omega$.
Further, 
the period of a monoid in $\ZG_p$ always divides~$p$.

\section{Characterizations of $\ZG$}
\label{sec:carac}

In this section, we present our characterizations of~$\ZG_p$ and of~$\ZG$, which we will use
to prove Theorem~\ref{thm:locality2}.
We will show that $\ZG_p$ is intimately linked to a congruence on words called the
$n,p$-congruence. Intuitively, two words are identified by this congruence if the
subwords of the \emph{rare letters} (occurring less than $n$ times) are the
same, and the numbers of occurrences of the \emph{frequent letters} (occurring
more than $n$ times) are congruent modulo~$p$. This is the standard technique of
\emph{stratification}, used also, e.g., in~\cite[Example~0.1]{therien85} or
in~\cite[Section 10.8]{almeida1994finite}.
Formally:

\begin{defi}[Rare and frequent letters, $n,p$-congruence]
\label{def:rare}
  Fix an alphabet $\Sigma$ and a word $w \in \Sigma^*$.
  Given a integer $n \in \NN$ called the \emph{threshold}, we call $a \in \Sigma$ \emph{rare} in~$w$ if $\card{w}_a \leq n$, and \emph{frequent} in~$w$ if $\card{w}_a > n$.
  The \emph{rare alphabet} is the (possibly empty) set $\{a \in \Sigma \mid \card{w}_a \leq n\}$.
We define the \emph{rare subword} $\restr{w}{\leq n}$ to be the subword of~$w$
  obtained by keeping only the rare letters of~$w$.

  For $n > 0$ and for any integer $p>0$ called the \emph{period}, the \emph{$n,p$-congruence} $\sim_{n,p}$
  is defined by writing $u\sim_{n,p}
v$ for 
$u, v \in \Sigma^*$ iff:
  \begin{itemize}
    \item The rare subwords are equal: $\restr{u}{\leq n} = \restr{v}{\leq n}$;
    \item The rare alphabets are the same: for all $a \in \Sigma$, we have $\card{u}_a > n$ iff $\card{v}_a > n$;
    \item The number of letter occurrences modulo~$p$ are the same: for all $a
      \in \Sigma$, we have that $\card{u}_a$ and $\card{v}_a$ are congruent
      modulo~$p$. (Note that we already know this for rare letters
      using the previous conditions.)
  \end{itemize}
\end{defi}

We first remark that two $n,p$-equivalent words are also $m,q$-equivalent for any
divisor $q$ of~$p$ and any $m \leq n$:
\begin{clm}
	\label{clm:multiple}
	For any alphabet $\Sigma$, for any $0 < m \leq n$ and $p, q > 0$ such
        that $q$ divides~$p$, 
        the $n,p$-congruence refines the $m,q$-congruence.
\end{clm}

Intuitively, it is less precise to look for exact subwords up to a lower
threshold and with a modulo that divides the original modulo. Here is the formal
proof:

\begin{proof}[Proof of Claim~\ref{clm:multiple}]
        We first assume that $p = q$.
        The claim is trivial for $m = n$, so we assume $n>m$.
	As $n>m$, if two words $u$ and $v$ have the same rare alphabet for $n$, then they have the same rare alphabet for $m$, because the number of occurrences of all rare letters for~$n$ is the same, so the same ones are also rare for~$m$.
        Furthermore, if $u$ and $v$ have the same rare subword for~$n$,
        the restriction of this same rare subword to the rare letters for~$m$
        yields the same rare subword for~$m$.
        Last, the number of occurrences of the frequent letters modulo $q$ is
        the same. Indeed, for the letters that were frequent for~$n$, this is
        the case because their number of occurrences is congruent modulo~$p =
        q$.
        For the letters that were not
        frequent for~$n$, this is because their number of occurrences has to be
        the same because the rare subwords for~$n$ were the same.

        We now assume that $q < p$. Let us now assume that $m = n$. Then the
        $n,p$-congruence refines the $m,q$-congruence because the rare
        alphabets and subwords must be equal (the threshold is unchanged), and
        the counts of the frequent letters modulo~$p$ determine these counts
        modulo~$q$ as $q$ divides $p$.

        To conclude in the general case, we know that the $n,p$-congruence
        refines the $n,q$-congruence which refines the $m,q$-congruence, so we
        conclude by transitivity.
\end{proof}

What is more, observe that $n,p$-congruences are a particular case of
$\ZG_p$-congruences:

\begin{clm}
  \label{clm:niszg}
  For any alphabet $\Sigma$ and $n > 0$ and $p > 0$,
  the $n,p$-congruence over $\Sigma^*$ is a $\ZG_p$-congruence.
\end{clm}

\begin{proof}
	Let $E$ be an equivalence class of the $n,p$-congruence, which we see as a
        language of~$\Sigma^*$, and let us show that $E$ is a language of~$\ZG_p$.
        Intuitively, $E$ defines a language where the number of occurrences of
        each letter modulo~$p$ is fixed, where the set of letters occurring
        $\leq n$ times (the rare letters) is fixed, and the
        subword that they achieve is also fixed.
        Formally, let $\Sigma = A \sqcup B$ the partition of~$\Sigma$ in rare and
        frequent letters for the class~$E$, let $u$ be the word over~$A^*$
        associated to the class~$E$, and let $\vec{k}$ be the $\card{B}$-tuple
        describing the modulo values for~$E$.

        We know that the singleton language $\{u\}$ is a language of~$\ZG_p$,
        because it is finite. Hence, the language $U = B^* u_1 \cdots B^* u_n
        B^*$ is also in $\ZG_p$, because it is the inverse of $\{u\}$ by
        the morphism that erases the letters of~$B$ and is the identity on~$A$.
        Similarly, the language $C$ of words of~$B^*$ where the number of
        occurrences of each letter 
        modulo~$p$ are as prescribed by~$\vec{k}$ and where every letter occurs
        at least $n$ times is a language of~$\ZG_p$, because it is commutative
        and $p$ is a multiple of the period of every group element.
        For the same reason, the language $C'$ of words of~$\Sigma^*$ whose
        restriction to~$B$ are in~$C$ is also a language of~$\ZG_p$, because it is
        the inverse image of~$C$ by the morphism that erases the letters of~$A$
        and is the identity on~$B$.
        Now, we remark that $E = C' \cap U$, so $E$ is in~$\ZG_p$, concluding the
        proof.
\end{proof}

The goal of this section is to show the following result. Intuitively, it states
that $\ZG_p$-congruences are always refined by a sufficiently
large $n,p$-congruence. Formally:
\begin{thm}
\label{thm:congruence}
For any $p > 0$, consider any $\ZG_p$-congruence $\sim$ over $\Sigma^*$ and consider the quotient
  $M \colonequals \Sigma^*/{\sim}$.
Let $n\colonequals(\card{M}+1)$.
Then the congruence $\sim$ is refined by the~$n,p$-congruence on~$\Sigma$.
\end{thm}

We first present some consequences of Theorem~\ref{thm:congruence}, and
then prove Theorem~\ref{thm:congruence}.

\subsection{Consequences of Theorem~\ref{thm:congruence}}

  The most important consequence of Theorem~\ref{thm:congruence} is a
characterization of the languages in~$\ZG$, which is
similar
to the one obtained by Straubing in~\cite{Straubing82} for the variety $\MNil$.
To define~$\MNil$, first define the operation $S^1$ on a semigroup~$S$: this is
the monoid obtained from~$S$ by adding an identity element $1$ to~$S$
if $S$ does not have one. Now, 
define a \emph{nilpotent semigroup} $S$ to be a semigroup
satisfying the equation $x^\omega y = y x^\omega = x^\omega$.
The variety $\MNil$ is generated by monoids of the form $S^1$ for
$S$ a nilpotent semigroup. Note that $\MNil \subseteq \ZG$, because the equation
of $\MNil$ implies the equation of~$\ZG$.
It was shown in~\cite{Straubing82} that the languages of $\MNil$ are generated
by \emph{disjoint
monomials},
that is, they are Boolean combinations of languages of the form $B^* a_1 B^* a_2 \cdots
a_k B^*$ with $B\cap \{a_1,\ldots, a_k\}=\emptyset$.

Our analogous characterization for~$\ZG$ is the following, obtained via
Theorem~\ref{thm:congruence}:

\begin{cor}
  \label{cor:zgcarac}
A language is in $\ZG$ if and only if it can be expressed as a finite union of languages of the form
$B^* a_1 B^* a_2 \cdots a_k B^* \cap K$
where $B \cap \{a_1, \ldots, a_k\} = \emptyset$ and $K$
is a regular commutative language.
\end{cor}

Equivalently, we can say that every language of~$\ZG$ is a
finite union of \emph{disjoint shuffles} of a singleton language (containing
only one word) and of a regular commutative language, where the \emph{disjoint
shuffle} operator interleaves two languages (i.e., it describes the sets of words
that can be achieved as interleavings of one word in each language) while
requiring that the two languages are on disjoint alphabets.
Intuitively, these characterizations are because the syntactic congruence of a
$\ZG$ language is a $\ZG$-congruence, and Theorem~\ref{thm:congruence} tells us
that it is refined by an $n,p$-congruence, whose classes can be expressed as
stated.
Let us formally prove Corollary~\ref{cor:zgcarac} using
Theorem~\ref{thm:congruence}:

\begin{proof}[Proof of Corollary~\ref{cor:zgcarac}]
  One direction is easy: if a language $L$ is of the prescribed form, then it is a
  Boolean combination of languages of $\MNil$ and regular commutative languages.
  These languages are in $\ZG$ and $\ZG$ is closed under Boolean operations, so
  indeed $L$ is in~$\ZG$.

  For the converse direction,
fix a language $L$ in~$\ZG$, and consider the syntactic congruence $\sim$
  of~$L$: it is a $\ZG$-congruence,
  more specifically a $\ZG_p$-congruence for some value $p>0$.
  By Theorem~\ref{thm:congruence},
there exists $n \in \NN$ such that $\sim$ is refined by a $n,p$-congruence~$\sim'$.
Now, by definition of the syntactic congruence, the set of words of~$\Sigma^*$ that are in~$L$ is a union of equivalence classes of~$\sim$, hence of~$\sim'$. This means that $L$ can be expressed as the union of the languages corresponding to these classes.

  Now, an equivalence class of the $n,p$-congruence $\sim'$ can be expressed as
  the shuffle of two languages: the singleton language containing the rare word
  defining the class, and the language that imposes that all
  frequent letters are indeed frequent (so the rare alphabet is as required)
  and that the modulo of their number of occurrences is as specified. The
  second language is commutative, and the disjointness of rare and frequent
  letters guarantees that the shuffle is indeed disjoint.

  Thus, we have shown that $L$ is a union of disjoint shuffles of a singleton
  language and a regular commutative language. The form stated in the corollary
  is equivalent, i.e., it is the shuffle of the singleton language $\{a_1 \cdots
  a_k\}$ and of the commutative language obtained by restricting~$K$ to the
  subalphabet~$B$.
\end{proof}

Corollary~\ref{cor:zgcarac} then implies a characterization of the variety of monoids $\ZG$.
To define it, we use the \emph{join} of two varieties $\V$ and $\W$, denoted by
$\V\join\W$, which is the variety of monoids
generated by the monoids of $\V$ and those of $\W$.
Alternatively, the join is the least variety containing both varieties. We then
have:
\begin{cor}
  \label{cor:zgmoncarac}
The variety $\ZG$ is generated by commutative monoids and monoids of the form $S^1$ with $S$
a nilpotent semigroup. In other words, we have: $\ZG = \MNil \lor \Com$.
\end{cor}
\begin{proof}
	Clearly $\ZG$ contains both $\Com$ and $\MNil$. Furthermore, by
        Corollary~\ref{cor:zgcarac}, any language in $\ZG$
	is a union of intersections of a language in $\MNil$ and a language in~$\Com$. Hence, it is
	in the variety generated by these two varieties of languages, concluding the proof.
\end{proof}

On a different note, we will also use Theorem~\ref{thm:congruence}
to show a technical result that will be useful later. It
intuitively allows us to regroup and move arbitrary elements:
\begin{cor}
\label{cor:zgcommutes}
	For any monoid $M$ in $\ZG$, letting $n \geq \card{M}+1$,
	for any element $m$ of~$M$ and elements $m_1, \ldots, m_n$ of~$M$, we have
	\[
		m \cdot m_1 \cdot m \cdot m_2 \cdot m \cdots m \cdot m_n \cdot m \cdot
                m_n \cdot m = m^{n+1} \cdot m_1\cdots m_n.
	\]
\end{cor}
\begin{proof}
        We consider the free monoid~$M^*$.
	Let $\eta:M^* \to M$ be the onto morphism defined  by
        $\eta(m) \colonequals m$ for all $m \in M$.
        Let $\sim$ be the congruence that~$\eta$ induces
        over $M^*$, i.e., for $u, v \in M^*$, we have $u\sim v$ if $\eta(u)=\eta(v)$.
	Remark that, as $M$ is in $\ZG$, more specifically in $\ZG_p$ for some
        $p>0$,
        the congruence $\sim$ is a
        $\ZG_p$-congruence by definition.
        Hence, by Theorem~\ref{thm:congruence}, $\sim$ is refined by a $n,p$-congruence where $n=\card{M}+1$.
        Now, consider the two words in the equation that we wish to show: they are
        words of~$M^*$.
        As the letter $m$ over is then a frequent letter,
        we know that the two words are indeed $n,p$-congruent, which concludes
        the proof.
\end{proof}

\subsection{Proof of Theorem~\ref{thm:congruence}}

Having spelled out the consequences of Theorem~\ref{thm:congruence}, we prove it
in the rest of this section. It crucially relies on a general result about~$\ZG$ that we will use
in several proofs, and which is shown by elementary equation manipulations:

\begin{lem}
\label{lem:distromega}
	Let $M$ be a monoid of~$\ZG$, let $\omega$ be the idempotent power, and let $x, y \in M$. Then we have: $(xy)^\omega = x^\omega y^\omega$.
\end{lem}

\begin{proof}
  We show this claim by showing that two equalities establishing that each side
  of the equation is equal to the same term, namely, $x^\omega y^\omega
  (xy)^\omega$. Let us first show the first equality:
\begin{equation}
	\label{eqn:distr1}
  (xy)^\omega = x^\omega y^\omega (xy)^\omega.
\end{equation}
  To show Equation~\ref{eqn:distr1}, remember that we have $(xy)^\omega = xy
  (xy)^{\omega-1}$ from the definition of $(xy)^{\omega-1}$ in the preliminaries. Now, as $(xy)^{\omega-1}$ is
  central, the right-hand side is equal to $x (xy)^{\omega-1} y$.
        By injecting an $(xy)^\omega$ in the latter, we obtain:
        \[(xy)^\omega = x (xy)^{\omega-1} (xy)^\omega y.\]
Applying this equality $\omega$ times gives:
  \[
    (xy)^\omega = (x (xy)^{\omega-1})^\omega (xy)^\omega y^\omega.
  \]
  Now, we can expand $(x(xy)^{\omega-1})^\omega$, commuting the
  $(xy)^{\omega-1}$ to regroup the $x$ into $x^\omega$ and regroup the
  $(xy)^{\omega-1}$ into
  $((xy)^{\omega-1})^\omega$ which is equal to~$(xy)^{\omega}$, so that the first factor of the
  right-hand side is equal to $x^\omega (xy)^\omega$.
  By commuting, we obtain Equation~\ref{eqn:distr1}.

The second equality is:
\begin{equation}
	\label{eqn:distr2}
	x^\omega y^\omega = x^\omega y^\omega (xy)^\omega.
\end{equation}
        To show Equation~\ref{eqn:distr2},
        note that we have $x^\omega y^\omega = x^{\omega - 1} x y^{\omega - 1} y$, so by
        the equation of $\ZG$ we get:
        \[
          x^\omega y^\omega = x^{\omega - 1} y^{\omega - 1} xy.
        \]
	Now, we have $x^{\omega - 1} y^{\omega - 1} = x^{\omega - 1} x^\omega y^{\omega - 1} y^\omega$,
        and by the equation of~$\ZG$ we have:
  \[x^{\omega-1} y^{\omega-1} = x^{\omega - 1} y^{\omega - 1} x^\omega
  y^\omega.\]
Inserting the second equality in the first, we have:
  \[x^\omega y^\omega =
  x^{\omega - 1} y^{\omega - 1} x^\omega y^\omega x y.\]
Now, applying this equality $\omega$ times gives
	$x^\omega y^\omega = (x^{\omega - 1} y^{\omega -1})^\omega x^\omega y^\omega (xy)^\omega$.
        As the first factor of the right-hand side is equal to $x^\omega
        y^\omega$, we get
 	$x^\omega y^\omega = x^\omega y^\omega (xy)^\omega$. This establishes
        Equation~\ref{eqn:distr2}.

From Equations~\ref{eqn:distr1} and~\ref{eqn:distr2}, we immediately conclude
  the proof.
\end{proof}

To continue with our proof of 
Theorem~\ref{thm:congruence},
thanks to Lemma~\ref{lem:distromega},
we can now show a kind of ``normal form'' for $\ZG$-congruences, by arguing that
any word can be rewritten to a word where frequent letters are moved to the
beginning of the word, without breaking equivalence for the $\ZG$-congruence. This relies
on Lemma~\ref{lem:distromega} and allows us to get to the notion of
$n$-equivalence. Specifically:

\begin{clm}
	\label{clm:moveend}
	Let $\sim$ be a $\ZG$-congruence on~$\Sigma$.
	Let $n \colonequals (\card{M}+1)$ where $M$ is the monoid associated to~$\sim$.
	Then, for all $w \in \Sigma^*$, for every letter $a \in \Sigma$ which is
        frequent in~$w$ (i.e., $\card{w}_a > n$), writing $w'$ the restriction
        of $w$ to~$\Sigma \setminus \{a\}$,
        and writing $w'' \colonequals a^{\card{w}_a} w'$,
        we have: $w \sim w''$.
\end{clm}

\begin{proof}
  Define $n \colonequals \card{M}+1$ as in the claim statement, and let $\mu:\Sigma^* \to
        M=\Sigma^*/{\sim}$ be the morphism associated to $\sim$. Remark that by definition, for any words $u,v$, we have $u\sim v$ iff $\mu(u) = \mu(v)$.

Let us take an arbitrary $w$ and $a \in \Sigma$ such that $a$ is frequent in~$w$.
We can therefore write $w = w_1 a w_2 a \cdots w_m a w_{m+1}$ with $m =
  \card{w}_a > n >  \card{M}$.
Furthermore, letting $x_l=\mu(w_1a w_2 a\cdots w_l a)$ for each $1 \leq l \leq m$, as
  $m > \card{M}$ we know by the pigeonhole principle that
there exist $1 \leq i < j \leq m$ such that $x_i = x_j$. Furthermore, we have $x_j = x_i z \mu(a)$ where
$z= \mu(w_{i+1} a \cdots w_j)$.
By applying the equation $\omega$ times, we have that $x_i = x_i (z \mu(a))^\omega$.

Now, by Lemma~\ref{lem:distromega}, we have $(z\mu(a))^\omega = z^\omega \mu(a)^\omega$.
This is equal to $z^\omega \mu(a)^\omega \mu(a)^\omega$, and by now applying 
  Lemma~\ref{lem:distromega} 
  in reverse we conclude that
$(z\mu(a))^\omega = (z\mu(a))^\omega  \mu(a)^\omega$.
Finally, we obtain $x_i = x_i (z\mu(a))^\omega = x_i (z\mu(a))^\omega \mu(a)^\omega = x_i \mu(a)^\omega$.

  Now, the equation of~$\ZG$ ensures that $\mu(a)^\omega$ is central,
  so we can commute it in~$\mu(w)$ and absorb all occurrences of~$\mu(a)$
  in~$\mu(w)$, then move it at the beginning, while keeping the same $\mu$-image.
  Formally, from $x_i = x_i \mu(a)^\omega$, we have
	\[
		\mu(w)= \mu(w_1) \mu(a) \cdots \mu(w_i) \mu(a) \mu(a)^\omega \mu(w_{i+1}) \mu(a) \cdots \mu(w_m)
  \mu(a) \mu(w_{m+1}),
	\]
   and we commute $\mu(a)^\omega$ to merge it with all
  $\mu(a)$ and then commute the resulting $\mu(a)^{\omega+\card{w}_a}$ to obtain
  $\mu(a)^{\omega+\card{w}_a} w'$ with $w'$ as defined in the statement of the
  claim.

  Now, remark that for any $i > \card{M}$, for any $x$ in $M$, we have $x^{\omega+i}
  = x^i$. This is because, letting $k \leq \card{M}$ be the idempotent power of $x$,
  we have $x^\omega = x^k$, hence $x^{\omega+i} = x^{k+i} = x^{i-k} x^k x^k =
  x^{i-k} x^k = x^i$. Note that $i-k > 0$ because $i > \card{M}$ and $k \leq
  \card{M}$, so
  $x^{i-k}$ is well-defined.

  By applying this to $x = \mu(a)$ and $i = \card{w}_a > \card{M}$,
  we deduce that $\mu(w) = \mu(a)^{\card{w}_a} \mu(w')$. 
  This establishes that
  $w \sim w''$ and concludes the proof.
\end{proof}

We can now conclude the proof of Theorem~\ref{thm:congruence}:

\begin{proof}[Proof of Theorem~\ref{thm:congruence}]
	Let $\sim$ be a $\ZG_p$-congruence on~$\Sigma^*$, $M$ its associated monoid, and fix $n
        \colonequals \card{M}+1$ as in the theorem statement.
Let $u$ and $v$ be two $n,p$-congruent words of~$\Sigma^*$, we need
to prove that they are indeed $\sim$-equivalent. Let $\Sigma' = \{a_1, \ldots,
  a_r\}$ be the subset of letters in~$\Sigma$ that are frequent in~$u$ (hence
  in~$v$, as they are $n,p$-congruent).
By successive applications of Claim~\ref{clm:moveend} for every frequent letter
  in~$\Sigma'$, starting with~$u$, we know that
  $u \sim a_1^{\card{u}_{a_1}} \cdots a_r^{\card{u}_{a_r}}\restr{u}{\leq n} $.
Likewise, we have
  $v \sim a_1^{\card{v}_{a_1}} \cdots a_r^{\card{v}_{a_r}}\restr{v}{\leq n} $.
  Now, for any $1 \leq i \leq r$, 
  the values $\card{u}_{a_i}$ and $\card{v}_{a_i}$ are greater than~$n$
  which is $\geq \card{M}$, and they are congruent modulo~$p$, which is a
  multiple of the period of~$a_i$, so we have $a_i^{\card{u}_{a_i}} \sim a_i^{\card{v}_{a_i}}$.
  We also know by definition of the $n,p$-congruence that $\restr{u}{\leq n} =
  \restr{v}{\leq n}$. By compositionality of~$\sim$, all of this establishes that $u \sim v$.
Thus, the $n,p$-congruence indeed refines the $\sim$-congruence, concluding the proof.
\end{proof}

\section{Defining $\ZGD$ and $\LZG$, and Result Statement}
\label{sec:zgli}

We have given our characterizations of~$\ZG$ and presented some preliminary
results. We now define
$\LZG_p$ and $\ZGpD$ and show that they are equal
(Theorem~\ref{thm:locality2}), and deduce the same for $\LZG$ and $\ZGD$.

\paragraph*{$\ZGpD$ and $\ZGD$.}
We denote by $\D$ the variety of the \emph{definite semigroups}, i.e.,
the semigroups satisfying the equation $y x^\omega = x^\omega$.
For $p>0$,
the variety of semigroups $\ZGpD$ is the variety generated by the
\emph{semidirect products}
of monoids in $\ZG_p$ and semigroups in~$\D$.
We recall for completeness the definition of the semidirect product operator,
even though 
we will not use it directly in this paper.
Given two semigroups $S$ and $T$, a \emph{semigroup action} of $S$ on $T$ is
defined by a map $\act\colon S\times T \to T$ such that
$\act(s_1, \act(s_2, t))= \act(s_1s_2, t)$ and $\act(s, t_1t_2)=\act(s,t_1) \act(s,t_2)$.
We then define the  product $\circ_\act$ on the set
$T\times S$ as follows: for all $s_1, s_2$ in $S$
and $t_1, t_2$ in $T$, we have: $(t_1, s_1)\circ_\act(t_2, s_2) \colonequals (t_1 \act(s_1, t_2), s_1 s_2).$
The set $T\times S$ equipped with the product $\circ_\act$ is a semigroup called the
\emph{semidirect product} of $S$ by $T$, denoted $T\circ_\act S$.

We then define $\ZGD$ as the variety of semidirect products of monoids in $\ZG$
and semigroups in~$\D$.

Remark that the ${*}$ operation is equivalent to the \emph{wreath product of
varieties}. We further note for the expert reader that
we could equivalently replace $\D$ by the variety $\Li$ of \emph{locally trivial
semigroups}.
We refer to~\cite{straubing1985finite} for a detailed presentation on this subject.

\paragraph*{$\LZG_p$ and $\LZG$.}
Last, we introduce the varieties $\LZG_p$ and $\LZG$.
For $p>0$,
the variety $\LZG_p$ is the variety of semigroups $S$
 such that, for every
idempotent~$e$ of~$S$, the subsemigroup $eSe$ of elements that can be written as
$ese$ for some $s\in S$  is in~$\ZG_p$.
Note that this subsemigroup is actually a monoid, called the \emph{local monoid}
of~$e$, and that its identity is~$e$.
The variety $\LZG$ is defined analogously but for~$\ZG$, and clearly $\LZG =
\bigcup_{p>0} \LZG_p$.
Remark that a semigroup is in $\LZG$
iff it satisfies the following
equation: for any $x$, $y$ and~$z$ in~$S$,
we have:
\[
	(z^\omega xz^\omega)^{\omega+1} (z^\omega yz^\omega) = (z^\omega yz^\omega) (z^\omega xz^\omega)^{\omega+1}.
\]
This incidentally shows by Reiterman's theorem that $\LZG$ is indeed a variety
of semigroups.

\paragraph*{Main result.}
Our main result, stated in Theorem~\ref{thm:locality2}, is that, for any $p>0$,
the varieties $\ZGpD$ and
$\LZG_p$ are actually equal. In particular, for $p = 1$, we get that $\MNilD =
\LMNil$.

Let us first remark that Theorem~\ref{thm:locality2} implies
Corollary~\ref{cor:locality} stated in the introduction, namely, that $\LZG =
\ZGD$:

\begin{proof}[Proof of Corollary~\ref{cor:locality}]
  Any semigroup in $\LZG$ is in
$\LZG_p$ for some~$p$, hence in $\ZGpD$ by locality of~$\ZGpD$, hence in $\ZGD$.
  Conversely, by definition of $\ZGD$ being a variety,
  each semigroup $S$ of $\ZGD$ is obtained by applying the quotient, subsemigroup, and
  product operators to semigroups $S_1, \ldots, S_k$ for some $k \in \NN$, where
  each $S_i$ is a semidirect product of a monoid in
  $\ZG_{p_i}$ for a certain $p_i > 0$ and of a semigroup in $\D$.
  Letting $p \colonequals \prod_{1 \leq i \leq k} p_i$, 
  all semigroups $S_1, \ldots, S_k$ are then in $\ZGpD$, and applying
  the operators then witnesses that $S$
  also belongs to $\ZGpD$.
  By locality of~$\ZGpD$, we know that $S$ belongs to $\LZGp$, hence to
$\LZG$.
\end{proof}

Thus, fixing the value $p>0$ for the rest of this paper, the only remaining
task is to prove Theorem~\ref{thm:locality2} on $\LZGp$ and $\ZGpD$. 
Note that we can freely use all equational results shown about monoids of $\ZG$
(e.g., Lemma~\ref{lem:distromega}), as
they also hold for monoids in~$\ZG_p$.

To prove Theorem~\ref{thm:locality2}, we will first present the general framework of Straubing's
delay theorem in the next section and show the easy inclusion $\ZGpD \subseteq
\LZG_p$,
before moving on with the rest of the proof.

\section{Straubing's Delay Theorem}
\label{sec:straubing}

To show our main result, we use Straubing's delay theorem from~\cite{straubing1985finite}.
We first give some prerequisites to recall this result.
To this end, let us first define a general notion of \emph{finite category}:

\begin{defi}
  A \emph{finite category} on a finite set of \emph{objects} $O$ defines, for every pair $(o,
  o') \in O$ of objects, a finite set $\CC_{o, o'}$ of \emph{arrows}. An arrow in $\CC_{o,
  o'}$ is said to be \emph{going from~$o$ to~$o'$}; we call $o$ the
  \emph{starting object} and $o'$ the \emph{ending object}.

  The overall set of arrows $\bigcup_{o,o' \in O} \CC_{o,o'}$ is equipped with a composition law: for any objects $o, o', o''$, for
  any arrows $a \in \CC_{o, o'}$ and $b \in \CC_{o', o''}$, the composition law
  gives us $ab$ which must be an arrow of~$\CC_{o, o''}$. (Note that the objects
  $o$, $o'$, and $o''$ are not necessarily pairwise distinct.)
  Further, this composition law must be associative. What is more, we require
  that for any object $o$, there exists an arrow in $\CC_{o, o}$ which is the identity
  for all elements that it can be combined with (note that these arrows are in
  particular unique).
\end{defi}

We now define the notion of \emph{category of idempotents} of a semigroup:

\begin{defi}[Category of idempotents]
Let $S$ be a semigroup. The \emph{category of idempotents} $S_\EE$ of~$S$ is the finite category defined as follows:
\begin{itemize}
   \item The objects of $S$ are the idempotents of~$S$.
   \item For any idempotents $e_1$ and $e_2$ and any element $x$ of $S$ such
     that $x \in e_1 S e_2$, we have an arrow labeled by $x$ going from~$e_1$
     to~$e_2$, which we will denote by $(e_1, x, e_2)$. Formally, $\CC_{e_1,
     e_2} = \{(e_1, x, e_2) \mid x \in e_1 S e_2\}$.
\end{itemize}
The composition law of the category is $(e_1, x, e_2) (e_2, y, e_3) = (e_1, xy,
  e_3)$.
 Note that it is clearly associative thanks to the associativity of the
  composition law on $S$. The identity element in $\CC_{e, e}$ for an
  idempotent~$e$ is simply $(e, e, e)$.
\end{defi}

Let us now study~$S_\EE$ in more detail.
We denote by
$B \colonequals \{ (e_1, x, e_2) \mid x \in e_1 S e_2 \}$ the set of arrows of
the category of idempotents.

A \emph{path} of~$S_\EE$ is a nonempty word of $B^*$ whose sequence of arrows is
\emph{valid}, i.e., the ending object of each arrow except the last one is equal to the starting object of the next arrow.
Because $S_\EE$ is a category, each path is equivalent to an element of the
category, i.e., composing the arrows of the path according to the composition
law of the category will give one arrow of the category, whose starting and
ending objects will be the \emph{starting object of the path} (i.e., that of the
first arrow) and the \emph{ending object of the path} (i.e., of the last arrow).
Two paths are \emph{coterminal} if they have the same starting and ending object.
Two paths $\pi_1$ and $\pi_2$ are \emph{$S_\EE$-equal} if they evaluate to the same
category element, which we write $\pi_1 \equiv \pi_2$. 
Note that if two paths are $S_\EE$-equal then they must be coterminal.
A \emph{loop} is a path whose starting and ending objects are the same.

A \emph{congruence} on~$B^*$ is an equivalence relation~$\sim$ over~$B^*$ that
satisfies compositionality as previously defined.
Note that the relation is also defined on
words of~$B^*$ that are not valid, i.e., that do not correspond to paths; the 
compositionality requirement also applies to such words.

\begin{defi}[Compatible congruence]
        A congruence $\sim$ on $B^*$ is said to be \emph{compatible} with~$S_\EE$
        if
        for any two coterminal paths $\pi_1$ and $\pi_2$ of~$S_\EE$ such that
        $\pi_1 \sim \pi_2$,
        we have $\pi_1 \equiv \pi_2$.
		In other words, $\sim$ is compatible with~$S_\EE$ iff, on words
                of~$B^*$ that are coterminal paths, it refines
                $S_\EE$-equality.
\end{defi}

Recall the notion of a $\ZG_p$-congruence from Section~\ref{sec:prelim}.
We are now ready to state Straubing's delay theorem, which was introduced
in~\cite{straubing1985finite} and rephrased in~\cite{Tilson}.
The theorem applies to any variety, but we state it specifically for~$\ZGp$ for
our purposes. The theorem gives us an alternative characterization of $\ZGpD$:

\begin{thm}[Straubing's delay theorem, Theorem 5.2 of~\cite{straubing1985finite}]\label{thm:straubing}
        A semigroup $S$ is in $\ZGpD$ iff,
		writing $S_\EE$ the category of idempotents of~$S$ and $B$ its set
                of arrows,
 		there exists a $\ZGp$-congruence on~$B^*$ which is compatible with~$S_\EE$.
\end{thm}

Using our notion of $n,p$-congruence, via
Claim~\ref{clm:niszg} and
Theorem~\ref{thm:congruence},
we rephrase it again:
\begin{cor}
\label{cor:straubing}
	A semigroup $S$ is in $\ZGpD$ iff, writing $S_\EE$ and $B$ as above,
 there exists an $n,p$-congruence on~$B^*$ which is compatible with~$S_\EE$.
\end{cor}

Before moving on to the full proof of our main theorem
(Theorem~\ref{thm:locality2}),
we conclude the section by noticing that the Straubing delay theorem implies the easy direction
of our result, namely, if $L$ is in $\ZGpD$ then $L$ is in $\LZGp$. This easy direction
follows directly from \cite{Tilson}, but let us provide a self-contained
argument:

\begin{clm}
  \label{clm:easy}
		We have $\ZGpD \subseteq \LZGp$.
\end{clm}

\begin{proof}
If $S$ is in $\ZGpD$, then by Theorem~\ref{thm:straubing}, there exists a
  $\ZGp$-congruence $\sim$ compatible with $S_\EE$.
Let us now show that $S$ is in $\LZGp$ by showing that,
for any idempotent $e$, the local monoid $eSe$ is in~$\ZGp$.
Let $e$ be an idempotent. By definition of~$S_\EE$, the local monoid $eSe$ is isomorphic to the subset of arrows of~$S_\EE$ going from~$e$ to~$e$, with their composition law.
Let us denote this subset by~$B_e$.
Define $\sim_e$ to be the specialization of the relation~$\sim$ to~$B_e$. Remark
  that $M=B_e^*/{\sim_e}$ is a submonoid of $B^*/{\sim}$,
  and is hence in $\ZGp$ because $B^*/{\sim}$ is, and $\ZGp$ is a variety.
 Remark that since all words in~$B_e^*$ are valid paths in~$S_\EE$, the local
  monoid $N\colonequals eSe$ defines a congruence $\sim_2$ over $B_e^*$ where two paths are equivalent if they evaluate to the same monoid element.
We know that $\sim$, hence $\sim_e$, refines this congruence~$\sim_2$. Hence, $N$ is a quotient of $M$.
  Thus, $eSe$ is a quotient of $B_e^*/{\sim_e}$, which is a submonoid of a
  monoid in~$\ZGp$, concluding the proof.
\end{proof}

In the rest of this paper, we show the much harder direction, i.e., if $S$ is in
$\LZGp$ then $S$ is in $\ZGpD$.
To prove this, using 
Corollary~\ref{cor:straubing}, it suffices to show:

\begin{clm}
  \label{clm:hard}
  Let $S$ be a semigroup of $\LZGp$, write $S_\EE$ its category of idempotents and
  $B$ the set of arrows of~$S_\EE$.
  There exists an integer $n > 0$ such that the $n,p$-congruence on~$B^*$ is compatible with~$S_\EE$.
\end{clm}

This result then implies, by our rephrasing of Straubing's result
(Corollary~\ref{cor:straubing}), that $S$ is in $\ZGpD$. So in the rest of this paper we prove Claim~\ref{clm:hard}.
The proof is structured in three sections. First, in Section~\ref{sec:choosen},
we carefully choose the threshold~$n$ of the $n,p$-congruence to be ``large
enough'' so that we can enforce a gap between the number of occurrences of the
\emph{rare} and \emph{frequent} arrows according to the $n,p$-congruence, also
ensuring a strong connectedness property on the category of idempotents.
Second, in Section~\ref{sec:auxpaths}, we show auxiliary results about paths in
the category of idempotents, using the $\ZG$ equation and the properties of the
threshold.
Third, in Section~\ref{sec:actualproof}, we conclude the proof, doing first an
outer induction on the total number of rare arrow occurrences, and then an inner
induction on the number of frequent arrows, using 
a so-called \emph{ear decomposition} of the category.

\section{Choosing the Threshold of the Congruence}
\label{sec:choosen}

In this section, we explain how to choose the threshold~$n$
to prove Claim~\ref{clm:hard}. Our threshold will
enforce a large enough ``gap'' between the total number of occurrences of the
rare letters and the number of occurrences of the frequent letters. This is
called a \emph{distant} threshold and will be useful in the rest of the
argument.

This section is split into two subsections.
In the first subsection, we formally define the notion of a
distant rare-frequent threshold, and show how such a threshold can be used to
find factors with no rare letters and a large number of copies of a frequent
letter.
We further explain that we can indeed
find a sufficiently distant rare-frequent threshold. This is a 
generic result on words and alphabets that does not depend on~$\ZG$ or
on the category of idempotents.
In the second subsection, we instantiate this result for
paths in the category of idempotents, and
explain how a distant rare-frequent threshold incidentally ensures
a strong connectedness property on the category. This last subsection
focuses on the category of idempotents, but is also generic in the sense that it
does not use the $\ZG$ equation. Hence, all of the present section applies to arbitrary
semigroups.

\subsection{Finding sufficiently distant thresholds}

We now give our formal definition of a \emph{distant} rare-frequent threshold:
\begin{defi}
\label{def:distant}
        For $\Sigma$ an alphabet, $u \in \Sigma^*$, and $m >0$,
        we say that an integer $n > 0$ is an \emph{$m$-distant rare-frequent threshold} for~$u$ if,
        letting $\Sigma_r \colonequals \{a \in \Sigma \mid \card{u}_a \leq n\}$
        be the rare alphabet of~$u$ for~$n$,
        then the \emph{total} number of occurrences of rare letters in~$u$ is less
        than a proportion $1/m$ of the threshold~$n$, minus one. Formally:
        $\sum_{a \in \Sigma_r} \card{u}_a \leq \frac{n}{m} - 1$.
\end{defi}

Note that if $n > 0$ is an $m$-distant rare-frequent threshold for~$u$, then it
is in particular an $m'$-distant rare-frequent threshold for any $0 < m' \leq m$.

If a word has a distant rare-frequent threshold, then any frequent letter can be
found in sufficiently many occurrences in some factor containing no rare letter.
This will be useful in pumping arguments, and is the motivation for the definition.
Formally:

\begin{lem}
  \label{lem:pump}
  For any alphabet $\Sigma$ and $m>0$, if a word $u \in \Sigma^*$ has an
  $m$-distant rare-frequent threshold, then for any frequent letter $a$, the
  word $u$ has a factor containing no rare letter and containing
  at least $m+1$ occurrences of~$a$.
\end{lem}

\begin{proof}
  As $a$ is a frequent letter, by definition, its number $\card{u}_a$ of
  occurrences in~$u$ is such that $\card{u}_a >n$. Let $r
  \colonequals \sum_{a \in \Sigma_r} \card{u}_a$ be the total number of occurrences of the
  rare letters. The definition of an $m$-distant rare-frequent threshold ensures
  that $r \leq \frac{n}{m}-1$. Thus, $r+1 \leq \frac{n}{m}$. Multiplying by~$m$,
  we obtain $(r+1)m \leq n$. Thus, $\card{u}_a > (r+1)m$. As there are $r$ rare
  letters in total, there are $r+1$ subwords between them containing no rare
  letter, so the inequality implies that one of them contains $>m$ occurrences
  of~$a$, i.e., $\geq m+1$ occurrences of~$a$. This concludes the proof.
\end{proof}

We now intuitively state the existence, for any~$m > 0$,
of an $m$-distant rare-frequent threshold that can be used for any word~$u$.
The precise claim is more complicated to phrase,
as we cannot pick one $n$ which can serve as a $m$-distant rare-frequent threshold
for any word~$u$. Indeed, for any choice of~$n$,
there will always be words~$u$ where the number of rare letter
occurrences happens to be close to~$n$. However, we can pick a large enough $n$
such that, given any word~$u$, we can pick some $0 < n' \leq n$ as an
$m$-distant rare-frequent threshold for~$u$. We will state the claim more
generally about picking a threshold for a pair of words $u_1$ and $u_2$.

Here is the formal claim. We repeat that this result is a general claim about
words, which is not specific to~$\ZG$.

\begin{lem}
\label{lem:choosengen}
        For any alphabet $\Sigma$ and $m>0$, there exists an integer $n>0$ ensuring the
        following: for any words $u_1, u_2 \in \Sigma^*$, there exists an
        integer~$n'$ with $0 < n' \leq
        n$ which is an $m$-distant rare-frequent threshold for~$u_1$ and for~$u_2$.
\end{lem}

Let us now prove Lemma~\ref{lem:choosengen} in the rest of this subsection.
We first show an abstract result capturing the essence of the underlying
pigeonhole principle argument:

\begin{clm}
\label{clm:gap}
  For any $d > 0$ and $m' > 0$,
  there exists an integer $n \geq m'$ ensuring the following:
  for any $d$-tuple $T$ of integers,
  there exists an integer~$n'$ with $m' \leq n' \leq n$ such that
  $\sum_{i \in F} T_i \leq \frac{n'}{m'}$,
  where $F =\{i \mid T_i \leq n'\}$.
\end{clm}

Intuitively, the value $m'$ will be computed from $m$ to
ensure the ``minus one'' gap in Definition~\ref{def:distant},
the ``dimension'' $d$ will be the cardinality of the
alphabet~$\Sigma$ (multiplied by 2 as we consider two words),
the integer $n$ is the threshold that we will choose, and $n'$ is the value that
we wish to obtain.
Let us prove Claim~\ref{clm:gap}:

\begin{proof}[Proof of Claim~\ref{clm:gap}]
	Let us take $n \colonequals (m'd)^{d+2}$, which ensures $n \geq
        m'$.
 The candidate values of~$n'$ that we will consider are intuitively the
  following:
  $m'd$, $(m'd)^2,\ldots,(m'd)^{d+2}$.
 Now take any $d$-tuple $T$.
  For any $1 \leq i \leq d+2$, let $R_i = \{i' \mid T_{i'} \leq (m'd)^i \}$ be
  the coordinates where~$T$ has a value $\leq (m'd)^i$. By definition, we have
  $\emptyset \subseteq R_1 \subseteq \cdots \subseteq R_{d+2} \subseteq \{1,
  \ldots, d\}$. 
 Applying the pigeonhole principle on the cardinalities of these sets, there
  are $1 \leq i < j \leq d+2$ such that $R_i = R_j$.

  Let us set $n' \colonequals (m'd)^j$. By
  construction, we have $n' \geq m'$. Now,
  consider the sum $\sum_{i' \in R_j} T_{i'}$.
  As $R_j = R_i$, we know that for every $i' \in R_j$, we have $T_{i'} \leq
  (m'd)^i $. Thus, the sum is at most
  $d$ times this value because $T$ is a $d$-tuple. Formally,
  $\sum_{i' \in R_j} T_{i'}$ is at most $d
  (m'd)^i $, hence it is $\leq \frac{(m'd)^{i+1} }{m'}$, so it is $\leq
  \frac{n' }{m'}$ by definition of 
  $n' = (m'd)^j$ because $i < j$. This concludes the proof.
\end{proof}

With Claim~\ref{clm:gap}, it is now easy to show Lemma~\ref{lem:choosengen}:

\begin{proof}[Proof of Lemma~\ref{lem:choosengen}]
  Let $n$ be the value given by Claim~\ref{clm:gap} when taking $d \colonequals
  2\card{\Sigma}$ and $m' \colonequals (m+1)m$.

  Let us consider any pair of words $u_1, u_2 \in \Sigma^*$.
  Let $T$ be the $d$-tuple of the letter occurrences of~$u_1$, followed by those
  of~$u_2$.
  The statement of Claim~\ref{clm:gap} ensures that there exists an integer~$n'$
  with $n' \geq m' > m > 0$ such that, when using $n'$ as the
  rare-frequent threshold, the total number of rare letters
  in~$u_1$ plus in~$u_2$ is
  $\leq \frac{n'}{m'}$.

  Now, let us show that $\frac{n'}{m'} \leq \frac{n'}{m}-1$. To
  this end, as $m' = m (m+1)$, first note that $\frac{n'}{m'} \leq \frac{n'}{m+1}$. To
  show now that $\frac{n'}{m+1} \leq \frac{n'}{m}-1$, let
  us evaluate $\frac{n'}{m}-1 - \frac{n'}{m+1}$ and check that it is nonnegative.
  This difference evaluates to $\frac{n' - m(m+1)}{m(m+1)}$. Now, we have $n'
  \geq m'$, so $n' \geq m(m+1)$, which is $\geq 0$. So indeed $\frac{n'}{m'}
  \leq \frac{n'}{m}-1$.

  Combining the conclusions of the two previous paragraphs, we obtain that
  the total number of rare letters $u_1$ plus in $u_2$ is $\leq \frac{n'}{m}-1$.
  So the same is true of the rare letters in~$u_1$, and of the rare letters
  in~$u_2$. By contrast, the frequent letters in~$u_1$ occur $> n'$ times by
  definition, and
  the same is true of the frequent letters in~$u_2$.
  Hence, by
  Definition~\ref{def:distant}, the value $n'$ is an $m$-distant rare-frequent
  threshold for~$u_1$ and for~$u_2$,
  concluding the proof.
\end{proof}

Thus, we have shown that, for any alphabet size and desired distance $m$, we can
pick a value $n$ ensuring that for any two words we can use some $0 < n' \leq n$
as an $m$-distant rare-frequent threshold for them.

\subsection{Choice of threshold and resulting properties}
In the rest of this section, we instantiate the generic argument of the previous
section to the choice of an arbitrary semigroup $S$ and its category of idempotents
$S_\EE$: we claim the existence of a suitably distant rare-frequent threshold,
and observe that it additionally ensures the strong connectedness of the
category. However, the argument is still generic in the sense
that it applies to an arbitrary semigroup $S$, even if it is not in~$\ZG$.

Let us first rephrase Lemma~\ref{lem:choosengen} to this setting by picking as
alphabet the arrows $B$ of the category of idempotents $S_\EE$ of~$S$ and by
choosing $m \colonequals \card{S}$. We immediately obtain:

\begin{clm}
\label{clm:choosenspec}
        Let $S$ be a semigroup and let
        $S_\EE$ be its category of idempotents.
        There exists an integer $n>0$ ensuring the following:
        for any paths $u_1, u_2$ in~$S_\EE$, there is $0 < n' \leq n$ which is a
        $\card{S}$-distant rare-frequent threshold for~$u_1$ and for~$u_2$.
\end{clm}

The rephrasing of Lemma~\ref{lem:pump} is:

\begin{clm}
  \label{clm:pumpspec}
  If a path $u$ in~$S_\EE$ has an $\card{S}$-distant rare-frequent threshold,
  then for any frequent arrow $a$, the
  path $u$ has a factor containing no rare arrow and containing
  at least $m+1$ occurrences of~$a$.
\end{clm}

We now close the section by observing that having an $m$-distant rare-frequent
threshold for a path~$w$, indeed having simply an $1$-distant rare-frequent
threshold for~$w$, enforces a strong connectedness property on the category.
Specifically, the set of frequent arrows in this path for this threshold must form a
so-called \emph{union of strongly connected components (SCCs)}:

\begin{defi}
\label{def:usccs}
	Given $S_\EE$ and a subset $B'$ of its set of arrows $B$, we say that $B'$ \emph{is a union of SCCs} if,
letting $G$ be the directed graph on the objects of~$S_\EE$ formed of the arrows
of~$B'$, then all connected components of~$G$ are strongly connected.
\end{defi}

\begin{clm}
	\label{clm:usccs}
	Fix $S$ and $S_\EE$ and $B$, let $w \in B^*$ be a path of~$S_\EE$, and let
        $n' > 0$ be a 1-distant rare-frequent threshold of~$w$.
Then the set of frequent arrows of~$w$ for~$n'$ is a union of SCCs.
\end{clm}

\begin{proof}
	Consider $G$ the directed graph  of Definition~\ref{def:usccs}.
	Let us assume by way of contradiction that $G$ has a connected component which is not strongly connected. This means that there exists an edge $(u, v)$ of~$G$ such that there is no path from~$u$ to~$v$ in~$G$. Consider any frequent arrow $a$ in~$S_\EE$ achieving the edge $(u, v)$ of~$G$.

        As $a$ is frequent in~$w$, we know that $a$ occurs strictly more than
        $n'$ times in~$w$, hence $w$ contains at least $n'$ \emph{return paths},
        i.e., paths from the ending object~$v$ of~$a$ back to the starting object~$u$ of~$a$. As there is no path from~$v$ to~$u$ in~$G$, each one of these paths must contain an arrow of~$B$ which is rare in~$w$.

  Hence, the total number of rare arrows in~$w$ is at least~$n'$. But the 1-distant rare-frequent threshold condition imposes that the total number of rare arrow occurrences in~$w$ is $\leq n'-1$. We have thus reached a contradiction.
\end{proof}

\section{The Loop Insertion and Prefix Substitution Lemmas}
\label{sec:auxpaths}

We now show two auxiliary results on the category of idempotents of $\LZG$
semigroups, to be used in the sequel. 
The first result is the \emph{loop insertion lemma}: it allows us to insert
any loop of frequent arrows to the power~$\omega$ without affecting equivalence.
The second is 
the \emph{prefix substitution lemma}: it allows us to replace a prefix of
frequent arrows by another, without affecting equivalence, up to inserting a loop later in the path.

The results shown in this section hold for semigroups in $\LZG$, hence for
those in $\LZGp$ for any $p>0$. Thus, let us fix a semigroup $S$ in $\LZG$.
Recall that $S_\EE$ denotes the category of idempotents of~$S$,
and denote by $B$ the set of arrows of~$S_\EE$.

Remember that, by definition of the category of idempotents~$S_\EE$, for any idempotent
$e$ of~$S$, the set of loops with starting and ending object $e$ under the
product law of the category forms a monoid which is isomorphic to the local
monoid $eSe$ of~$S$. 
As we know that $S$ is in~$\LZG$, all its local monoids are
in~$\ZG$. Hence the $\ZG$ equation immediately applies to loops in the 
category of idempotents, which we will often use in the results of this section:
\begin{clm}
\label{clm:c0}
	Let $x$ and $y$ be two coterminal loops of~$S_\EE$, let $k \in \ZZ$, and
        let $\omega$ be an idempotent power of~$S$. We have:
$x^{\omega+k} y \equiv y x^{\omega+k}$.
\end{clm}

\subsection{Loop insertion lemma}
The loop insertion lemma allows us, when we have a sufficiently distant rare-frequent
threshold~$n'$,
to insert any arbitrary loop raised to the power $\omega$ without changing the
category element to which a path evaluates.
We also show that this change does not affect $n',\omega$-equivalence:
this will suffice to preserve $n',p$-equivalence (where $p$ is the period
of~$S$),
which will be useful later (we come back to this at the beginning of
Section~\ref{sec:actualproof}).
Formally:

\begin{lem}[Loop insertion lemma]
    \label{lem:spawnloop}
    Let $\pi$ be a path, and assume that $n'$ is an $\card{S}$-distant rare-frequent threshold
    for~$\pi$.
    Let $\pi = rt$ be a decomposition of~$\pi$ (with $r$ or $t$ possibly empty),
    let $o$ be the object between~$r$ and~$t$ (i.e., the ending object
    of~$r$, or the starting object of~$t$ if~$r$ is empty),
    and let $\pi'$ be a loop on~$o$ that only uses frequent arrows.
  Then $\pi \equiv r (\pi')^{\omega} t$.
  (Note that these two paths are also $n',\omega$-equivalent by construction.)
\end{lem}

  In the rest of this subsection, we show Lemma~\ref{lem:spawnloop}.
  We first rephrase the claim to the following auxiliary result:

  \begin{clm}
    \label{clm:auxiliary}
    Let $\pi$ be a path, and assume that $n'$ is an $\card{S}$-distant rare-frequent
    threshold for~$\pi$.
    Let $\pi = rt$ be a decomposition of~$\pi$ and $o$ be the object between~$r$
    and~$t$.
    Let $X$ be the set of elements of the local monoid on~$o$ that can be
    achieved as some loop $q_x$ of frequent arrows raised to the power $\omega$: formally, $X$ is
    the set of elements $x \in S$ such that there is a loop $q_x$ of frequent
    arrows such that $q_x^{\omega}$ evaluates to $(o, x, o)$, noting that this
    implies that $x$ is idempotent.
    Then letting $q \colonequals \prod_{x \in X} q_x^{\omega}$, we have
    that $\pi \equiv r q t$.
    (Note that these two paths are also $n',\omega$-equivalent by construction.)
  \end{clm}

  Intuitively, in this claim, $X$ stands for the set of elements of~$S$ that we
  can achieve as a loop on~$o$ raised at the power $\omega$ and using frequent arrows
  only. The claim states that inserting loops of this form at~$o$ to achieve all
  such elements~$x$ will preserve equivalence.

  We first explain why Claim~\ref{clm:auxiliary} implies Lemma~\ref{lem:spawnloop}. Indeed, when taking $p
  = rt$, letting $o$ be the object between $r$ and~$t$,
  and taking $r (\pi')^{\omega} t$, 
  the sets $X$ defined in Claim~\ref{clm:auxiliary} after the occurrence of~$r$ in
  $p=rt$ and $r (\pi')^\omega t$ will be the same for both
  paths (because $X$ only depends on~$o$).
  Thus, Claim~\ref{clm:auxiliary} implies that there is a loop $q$ such that $rt \equiv
  rqt$ and  $r(\pi')^{\omega}t \equiv rq(\pi')^{\omega}t$.
 Now, as $(\pi')^{\omega}$ must correspond to an arrow of the form $(o, x, o)$ for $x \in X$, it must be the same idempotent as one of the idempotents achieved by one of the loops in the definition of~$q$, and
  as the local monoid is in~$\ZG$ these idempotents commute and $q \equiv
  q(\pi')^{\omega}$.
  Hence, we have
  $rqt \equiv rq(\pi')^{\omega}t$.
  We know that $rqt \equiv rt$, and $rq(\pi')^{\omega}t \equiv
  r(\pi')^{\omega}t$.
  Thus we obtain $rt \equiv r(\pi')^{\omega}t$.
  Thus,
  Lemma~\ref{lem:spawnloop} is proved once we have shown Claim~\ref{clm:auxiliary}.

  Hence, all that remains is to show Claim~\ref{clm:auxiliary}. We will do so by
  establishing a number of claims.

We first show that, for any frequent arrow $x$, we can
insert some loop of the form $xu$ where $u$ is a return path using only frequent
arrows, while preserving equivalence.
This uses the notion of distant rare-frequent
threshold; specifically, this is where we perform the pumping made possible by
Claim~\ref{clm:pumpspec}.

  \begin{clm}
    \label{clm:prelimb}
    Let $\pi = rt$ be a path
    with an $\card{S}$-distant rare-frequent threshold,
    let $o$ be the object between $r$ and $t$,
    and let $x$ be any frequent arrow starting at~$o$.
    Then we have $\pi \equiv r(xu)^\omega t$ for some return path~$u$ using only
    frequent arrows.
  \end{clm}

  \begin{proof}
    As $x$ is a frequent arrow and $\pi$ has an $\card{S}$-distant threshold, we know by
    Claim~\ref{clm:pumpspec} that $\pi$ contains a factor
    $\rho$ that contains only frequent arrows and contains $k > \card{S}$ occurrences of~$x$.
    This provides a decomposition of $\rho$ in the form:
$\rho = \rho_1 x \rho_2 x \cdots \rho_k x s$.

By the pigeonhole principle, there exists $i<j$ such that
$\rho_1 x \cdots \rho_i x \equiv \rho_1 x \cdots \rho_j x$.
Hence, iterating, we obtain:
\[
\rho_1 x \cdots \rho_j x \equiv
\rho_1 x \cdots \rho_i x (\rho_{i+1} x \cdots \rho_j x)^\omega.
\]
Moving the $\omega$, we get:
\[
\rho_1 x \cdots \rho_j x \equiv
  \rho_1 x \cdots \rho_{i-1} x \rho_i (x \rho_{i+1} x \cdots \rho_j)^\omega x.
\]
This proves that $\rho$ and $h(xu)^\omega g$ achieve the same category element
    when taking $u \colonequals \rho_{i+1} x \cdots \rho_j$, $h \colonequals \rho_1 x
    \cdots \rho_{i-1} x \rho_i$
  and $g \colonequals x \rho_{j+1} x \cdots \rho_k x s$.
    By compositionality, $\pi$ and $h' (xu)^\omega g'$ achieve the same
    category element, where $h'$ is the part of~$\pi$ preceding~$\rho$
    followed by~$h$, and $g'$ is $g$ followed by the part of~$\pi$
    following~$\rho$.

  Now, recall that we must show the result for our decomposition $\pi = rt$,
    where the object~$o$ between $r$ and $t$ is the starting object of the
    arrow~$x$ and the ending object of~$h$.  
Either $h'$ is a prefix of $r$, or vice-versa. Assume first that we are in the
    first case, so $h'=r w$ for some path~$w$. Then, $w$ and $(xu)^\omega$
belong to the local monoid of~$o$ which is in $\ZG$. Since idempotents commute
with all elements, we have $w(xu)^\omega \equiv (xu)^\omega w$ establishing that
$r w (xu)^\omega g' \equiv r (xu)^\omega wg' \equiv r(xu)^\omega t$ since
    $wg'=t$.
    The other case is symmetrical. This concludes the proof of
    Claim~\ref{clm:prelimb}.
  \end{proof}

We then prove a generalization of the previous claim, going from a single
frequent arrow to an arbitrary path of frequent arrows:

\begin{clm}
  \label{clm:prelim2}
    Let $\pi = rt$ be a path
    with an $\card{S}$-distant rare-frequent threshold.
    Let $o$ be the object between $r$ and $t$.
    Let $h$ be any path starting at~$o$ which
    only uses frequent arrows.
  Then we have  $\pi \equiv r (hg)^\omega t$ for some return path~$g$ using only
  frequent arrows.
\end{clm}

\begin{proof}
We show the claim by induction on the length of~$h$.
  The base case of the induction, with $h$ of length $0$, is trivial with $g$ also having length~$0$.

For the inductive claim, write $h = h' a$. Intuitively, we will insert a loop
  starting with the path $h'$, then insert a loop starting with the
  arrow~$a$ within that loop, and then recombine.

Formally,
by induction hypothesis, there exists a $g'$ using only frequent arrows 
  such that:
  \[\pi \equiv r (h'g')^\omega t.\]
  Furthermore, by applying  Claim~\ref{clm:prelimb} to the decomposition
  $r'=rh'$ and $t'=g'(h'g')^{\omega-1}t$ and with the frequent arrow~$a$ we get a return path $u$ using only
  frequent arrows 
  such that:
  \[ r (h'g')^\omega t \equiv r h' (au)^\omega g' (h'g')^{\omega-1}t\]
So, iterating the~$\omega$ power, and combining with the preceding equation, we
get:
  \[\pi \equiv r h' ((au)^\omega)^\omega g' (h'g')^{\omega-1}t.\]
Now, by applying $\omega-1$ times Claim~\ref{clm:c0} to each $(au)^\omega$ except the first and to each loop going from after this $(au)^\omega$ to the position between an occurrence of~$h'$ and~$g'$,
we get that:
  \[
    \pi \equiv r h' (au)^\omega g' (h'(au)^\omega g')^{\omega-1}t.\]
Note the right-hand side is equal to:
$r (h'(au)^\omega g')^\omega t$.
  So we have shown:
  \[\pi \equiv r (h' (au)^\omega g')^\omega t.\]
So this establishes the inductive claim by taking $g \colonequals u (au)^{\omega-1} g'$.
\end{proof}

We are interested in the specialization of this result when the path of
frequent arrows to insert is a loop. In this case, the return path is also a
loop. Formally, the specialization is the following:
\begin{cor}
  \label{cor:aux2}
    Let $\pi = rt$ be a path
    with an $\card{S}$-distant rare-frequent threshold $n'$,
    let $o$ be the object between $r$ and $t$,
    and let $q$ be a loop on~$o$ using only frequent arrows.
  We have that $r t \equiv r q^{\omega} (q')^{\omega} t$ for some loop $q'$
  on~$o$ using only frequent arrows
  (note that the two are also $n',\omega$-equivalent).
\end{cor}

\begin{proof}
  We use Claim~\ref{clm:prelim2} with $h \colonequals q$. This gives us the
  existence of a return path~$g$ using only frequent arrows, which is then also a loop on~$o$, such that
$rt \equiv r (q g)^\omega t$.
Now, applying Lemma~\ref{lem:distromega} to the local monoid on object~$o$, we know that this evaluates to the same category element as
$r q^\omega g^\omega t$.
  Taking $q' \colonequals g$ concludes the proof of Corollary~\ref{cor:aux2}.
\end{proof}

The only step left is to argue that Corollary~\ref{cor:aux2} implies our rephrasing
(Claim~\ref{clm:auxiliary}) of the loop insertion lemma
(Lemma~\ref{lem:spawnloop}).
To do this, let $\pi = rt$ be the path,
let $o$ be the object between $r$ and~$t$,
and let $X$ be the set of idempotents definable from frequent arrows.
Let us write $X = \{x_1,\ldots, x_k\}$ with $k = \card{X}$.
For each $1 \leq i \leq k$, take a loop $q_i$ on~$o$ consisting only of frequent arrows 
that achieves $x_i$, i.e., which ensures that $q_i^\omega$ evaluates to $(o, x_i, o)$.
Now, we apply Corollary~\ref{cor:aux2} $k$ times with $q$ being each of the
loops
$q_1, \ldots, q_k$. We get the following, where $q'_1, \ldots, q'_k$ are the
loops $q'$ on~$o$ of frequent arrows obtained by the statement of Corollary~\ref{cor:aux2}:
\[
  r t \equiv r q_1^\omega (q'_1)^\omega \cdots q_k^\omega (q_k')^\omega t
\]
Note that the left-hand-side and right-hand-side are also
$n',\omega$-equivalent.

Now, since the $q'_i$ are loops on~$o$ consisting of frequent arrows, each 
$(q'_i)^\omega$ evaluates to $(o, x, o)$ for some $x \in X$.
As the elements of~$X$ are in the local monoid of~$o$ which is in~$\ZG$, 
commuting the loops using Claim~\ref{clm:c0}, 
we can combine each $(q'_i)^{\omega}$ with some $q_j^\omega$ such that
$(q'_i)^\omega \equiv q_j^\omega$, hence $(q'_i)^\omega q_j^\omega =
q_j^\omega$.
Thus, we
get that $rt$ is $n',\omega$-equivalent to, and evaluates to the same category
element as, the path:
\[
  r \prod_{x\in X} q_x^{\omega} t.
\]
This concludes the proof of Claim~\ref{clm:auxiliary}, and thus establishes our
desired result, Lemma~\ref{lem:spawnloop}.

\subsection{Prefix substitution lemma}
The prefix substitution lemma allows us to change any prefix of frequent arrows of a
path, up to inserting a loop of frequent arrows elsewhere:
\begin{lem}
    \label{lem:prefixfrequent}
    Let $\pi = x r y$ be a path, and assume that $n'$ is an $\card{S}$-distant rare-frequent threshold
    for~$\pi$.
    Let $x'$ be a path coterminal with $x$. Assume that every arrow in $x$
    and in $x'$ is frequent.
    Assume that some object $o$ in the SCC of
  frequent arrows of the starting object of $r$ occurs again in $y$, say as the
  intermediate object of $y = y_1 y_2$.
  Then there exists $y' = y_1 y'' y_2$ for some loop $y''$ consisting only of
  frequent arrows such that $\pi \equiv x' r y'$ and such that $\pi$ and $x' r
  y'$ are $n',\omega$-equivalent.
\end{lem}

This lemma uses Claim~\ref{clm:usccs} to argue that
frequent arrows are a union of SCCs. Its proof relies on the loop insertion
lemma (Lemma~\ref{lem:spawnloop}), but with extra technical work using the $\ZG$
equation.

To prove Lemma~\ref{lem:prefixfrequent}, we will first show
that frequent loops can be ``recombined'' without changing the category image,
simply by equation manipulation:
\begin{clm}
\label{clm:c1b}
        For $x$, $x'$ two coterminal paths in $S_\EE$ and $y$, $y'$ coterminal paths in $S_\EE$ such that $xy$ and $x'y'$ are valid loops, we have:
                $(xy)^\omega (x'y')^\omega \equiv (xy')^\omega (x'y)^\omega
                (xy)^\omega (x'y')^\omega$.
\end{clm}

\begin{proof}
Let us first show that:
\begin{equation}
\label{eqn:e1}
(xy)^\omega (x'y')^\omega  \equiv x' y (xy)^{\omega-1} (x'y')^{\omega-1} x y'
\end{equation}
To show Equation~\ref{eqn:e1}, first rewrite $(xy)^\omega$ as $x (yx)^{\omega-1} y$ and likewise for $(y'x')^\omega$, to get:
\[
(xy)^\omega (x'y')^\omega \equiv x (yx)^{\omega-1} y x' (y'x')^{\omega-1} y'
\]
Then, we use Claim~\ref{clm:c0} to move $(yx)^{\omega-1}$, so the above
  evaluates to the same category element as:
\[
x y x' (yx)^{\omega-1} (y'x')^{\omega-1} y'
\]
We again rewrite $(yx)^{\omega-1}$ to $y (xy)^{\omega-2} x$, yielding:
\[
x y x' y (xy)^{\omega-2} x (y'x')^{\omega-1} y'
\]
We again use Claim~\ref{clm:c0} to move $(xy)^{\omega-2}$, merge it with the prefix $xy$, and move it back to its place, yielding:
\[
x' y (x y)^{\omega-1} x (y'x')^{\omega-1} y'
\]
We rewrite $(y'x')^{\omega-1}$ to $y' (x' y')^{\omega-2} x'$, yielding:
\[
x' y (x y)^{\omega-1} x y' (x'y')^{\omega-2} x' y'
\]
Again by Claim~\ref{clm:c0}, we can merge $(x'y')^{\omega-1}$ with $x' y'$ and move it to finally get:
\[
x' y (x y)^{\omega-1} (x'y')^{\omega-1} x y'
\]
Thus, the left-hand side of Equation~\ref{eqn:e1} evaluates to the same category
  element as the right-hand side, and we have shown Equation~\ref{eqn:e1}.

Now, we have $(xy)^{\omega} (x'y')^\omega \equiv (xy)^{2\omega}
  (x'y')^{2\omega}$, so using Claim~\ref{clm:c0} again, Equation~\ref{eqn:e1} gives:
  \[
    (xy)^\omega (x'y')^\omega  \equiv x' y  (xy)^\omega (x'y')^\omega (xy)^{\omega-1} (x'y')^{\omega-1} x y'
\]
  We can now use Equation~\ref{eqn:e1} to replace $(xy)^\omega (x'y')^\omega$ by
  the right-hand side of Equation~\ref{eqn:e1} and use Claim~\ref{clm:c0} to
  commute, yielding:
  \[
    (xy)^\omega (x'y')^\omega  \equiv (x' y)^2
    (xy)^{\omega-2} (x'y')^{\omega-2}
    (x y')^2
\]
  By definition we have $(xy)^{\omega-2} \equiv (xy)^{\omega-2} (xy)^\omega$ and
  $(x'y')^{\omega-2} \equiv (x'y')^\omega (x'y')^{\omega-2}$. Injecting these in
  the equation above, we get:
  \[
    (xy)^\omega (x'y')^\omega  \equiv (x' y)^2
    (xy)^{\omega-2} \underbrace{(xy)^\omega (x'y')^\omega}_{(\star)}
 (x'y')^{\omega-2}
    (x y')^2
\]
  Note that $(\star)$ is now equal to the left-hand-side of the equation. 
  Substituting $\omega$ times the right-hand-side into $(\star)$, we obtain:
  \[
(xy)^\omega (x'y')^\omega \equiv (x'y)^\omega (xy)^\omega (x'y')^\omega (xy')^\omega
\]
As these elements commute (thanks to Claim~\ref{clm:c0}), we have shown the desired equality.
\end{proof}

We will now extend this result to show that we can change
the initial part of a path, even if it is not a loop, provided that there is a
coterminal path under an~$\omega$-power with which we can swap it.

\begin{clm}
\label{clm:c3}
        For $x$, $x'$ two coterminal paths in $S_\EE$ and $y$, $y'$ coterminal paths in $S_\EE$ such that $xy$ and $x'y'$ are valid loops, and for any path $t$ coterminal with~$y$, the following equation holds:
                $x t (xy)^\omega (x'y')^\omega \equiv x' t (xy)^\omega xy' (x'y')^{\omega-1}.$
\end{clm}

We establish this again by equation manipulation.

\begin{proof}[Proof of Claim~\ref{clm:c3}]
We apply Claim~\ref{clm:c1b} to show the following equality about the left-hand side:
                \[x t (xy)^\omega (x'y')^\omega \equiv x t (xy')^\omega (x'y)^\omega (xy)^\omega (x'y')^\omega\]
By commutation of $(x'y)^\omega$ thanks to Claim~\ref{clm:c0}, the right-hand
  side evaluates to the same category element as:
                \[ (x'y)^\omega x t (xy')^\omega (x'y')^\omega (xy)^\omega\]
By expanding $(x'y)^\omega = x'(yx')^{\omega-1}y$, we get:
                \[x' (yx')^{\omega-1} y x t (xy')^\omega (x'y')^\omega (xy)^\omega\]
By commutation of $(xy)^\omega$ and expanding it to~$x (yx)^{\omega-1} y$, we get:
                \[x' (yx')^{\omega-1} y x (yx)^{\omega-1} y x t (xy')^\omega (x'y')^\omega\]
Combining $(yx)^{\omega-1}$ with what precedes and follows, we get:
                \[x' (yx')^{\omega-1} (yx)^{\omega+1} t (xy')^\omega (x'y')^\omega\]
By expanding $(x'y')^\omega = x' (y'x')^{\omega-1} y'$, and commuting $(yx')^{\omega-1}$ and $(yx)^{\omega+1}$, we get:
                \[x' t (xy')^\omega x' (yx')^{\omega-1} (yx)^{\omega+1} (y'x')^{\omega-1} y'\]
Now, we have $x' (yx')^{\omega-1}  = (x'y)^{\omega-1} x'$, so we get:
                \[x' t (xy')^\omega (x'y)^{\omega-1} x' (yx)^{\omega+1} (y'x')^{\omega-1} y'\]
Commuting $(y'x')^{\omega-1}$ and doing a similar transformation, we get:
                \[x' t (xy')^\omega (x'y)^{\omega-1} (x' y')^{\omega-1} x' (yx)^{\omega+1} y'\]
Now, expanding $(yx)^{\omega+1}$, we get:
                \[x' t (xy')^\omega (x'y)^{\omega-1} (x' y')^{\omega-1} x' y (xy)^{\omega} x y'\]
Commuting $(x'y)^{\omega-1}$ and merging it with $x' y$, we get:
                \[x' t (xy')^\omega (x' y')^{\omega-1} (x'y)^\omega (xy)^{\omega} x y'\]
Note that $(x' y')^{\omega-1} \equiv (x' y')^{\omega} (x' y')^{\omega-1}$, so applying commutation we get:
                \[x' t (xy')^\omega (x' y')^{\omega} (x'y)^\omega (xy)^{\omega} (x' y')^{\omega-1} x y'\]
  Now, applying Claim~\ref{clm:c1b} in reverse (using commutation again), we can obtain:
                \[x' t (x' y')^{\omega} (xy)^{\omega} (x' y')^{\omega-1} x y'\]
Then commuting $(x'y')^{\omega}$ and merging it yields:
                \[x' t (xy)^{\omega} (x' y')^{\omega-1} x y'\]
A final commutation of $(x' y')^{\omega-1}$ yields the desired right-hand side,
  and we have preserved equivalence in the category of idempotents, establishing the result.
\end{proof}

With Claim~\ref{clm:c3} in hand, and using Lemma~\ref{lem:spawnloop}, we can now prove
Lemma~\ref{lem:prefixfrequent}:

\begin{proof}[Proof of Lemma~\ref{lem:prefixfrequent}]
  As $n'$ is an $\card{S}$-distant rare-frequent threshold, it is in particular
  a 1-distant rare-frequent threshold, so we know by
  Claim~\ref{clm:usccs} that the frequent arrows occurring in~$\pi$
  are a union of SCCs. Thus, there is a
return path $s$ for $x'$ (i.e., $x' s$ is a loop, hence $x s$ also is) where
$s$ only consists of frequent arrows.

By our hypothesis on the starting object of~$r$, we can decompose $y = y_1 y_2$
  such that the ending object of~$y_1$ and starting object of~$y_2$ is an
  object~$o$ which is in the SCC of frequent arrows of the
  starting object $o'$ of~$r$.
  Let $\rho_1$ be any path of frequent arrows from $o$ to~$o'$, and $\rho_2$ be
  any path of frequent arrows from~$o'$ to~$o$.
  Now, take $\pi'$ to
  be the loop $\rho_1 s (xs)^\omega(x's)^\omega x \rho_2$: note that all arrows of~$\pi'$ are
  frequent. Hence,
  by Lemma~\ref{lem:spawnloop},
  $xry$ evaluates to the same category element as, and is $n',\omega$-equivalent to,
\[
  x r y_1 (\pi')^\omega y_2 = x r y_1 (\rho_1 s (xs)^\omega(x's)^\omega x \rho_2)^{\omega} y_2
\]
By unfolding the power $\omega$, we get the following:
\[
x r y_1 (\pi')^\omega y_2
  = x r y_1 
  \rho_1 s (xs)^\omega(x's)^\omega x \rho_2
  (\rho_1 s (xs)^\omega(x's)^\omega x \rho_2)^{\omega-1} y_2
= x r y_1 \rho_1 s (xs)^\omega(x's)^\omega x z
\]
where we write $z \colonequals \rho_2 (s (xs)^\omega(x's)^\omega x)^{\omega-1} y_2$ for
  convenience.
We can therefore apply Claim~\ref{clm:c3} to obtain that:
  \[(x (r y_1 \rho_1 s)) (xs)^\omega(x's)^\omega x z
  \equiv (x' (r y_1 \rho_1 s)) (xs)^\omega (x s) (x's)^{\omega-1} x z.\]
What is more, these two paths are clearly $n',\omega$-equivalent, as they only
  differ in terms of frequent arrows (all arrows in $x$ and $x'$ being frequent)
  and the number of these arrows modulo~$\omega$ is unchanged by the transformation.
This path is of the form given in the statement, taking $y' \colonequals
  y_1 \rho_1  s
  (xs)^\omega (x s) (x's)^{\omega-1} x z$ from which we can extract the right
  $y''$. This concludes the proof.
\end{proof}

\section{Concluding the Proof of the Main Result (Theorem~\ref{thm:locality2})}
  \nosectionappendix
\label{sec:actualproof}

We are now ready to prove the second direction of Theorem~\ref{thm:locality2},
namely Claim~\ref{clm:hard}.
Let us fix the semigroup $S$ in $\LZGp$, write $S_\EE$ its category of
idempotents, and write $B$ for the set of arrows of~$S_\EE$.
Let $p'$ be the period of~$S$: we know that $p'$ divides the
idempotent power $\omega$ of~$S$. Further,
$p'$ divides~$p$: this is
because, for any group element $x$ of~$S$, we have that $x$ belongs to the local
monoid $x^\omega S x^\omega$, hence it belongs to $\ZG_p$ and its period divides
$p$, thus $p'$ is a multiple of the periods of all group elements, hence it
divides~$p$.

Let~$n$ be given by Claim~\ref{clm:choosenspec}.
Our goal is to show that the $n,p$-congruence on~$B^*$ is compatible with~$S_\EE$.
We will in fact show the same for the $n,p'$-congruence, which is coarser than
the $n,p$-congruence because $p'$ divides $p$ (using Claim~\ref{clm:multiple}), so suffices to establish the result.

To do so, we will show that two coterminal paths that are
$n,p'$-equivalent must evaluate to the same category element, by two nested
inductions. We first explain the outer induction, which is 
on the number of rare arrows in the paths, for an $\card{S}$-distant
rare-frequent threshold $n' \leq n$ chosen from the two paths that we consider
(via Claim~\ref{clm:choosenspec}).
The threshold $n'$ is not fixed from the beginning but it depends on the two
paths considered, which is why we will not fix it in the inductive claim.
However, we will only apply
the claim to one value of~$n'$ chosen from the initial paths, i.e., it will not change during the induction.

Formally, we show by induction on the integer $r$ the following: 

\begin{clm}[Outer inductive claim on~$r$]
  \label{clm:outer}
For any two paths $\pi_1$ and $\pi_2$,
for any $\card{S}$-distant rare-frequent threshold $n'$ for $\pi_1$ and $\pi_2$,
if $\pi_1$ and $\pi_2$ are $n',p'$-equivalent and contain $r$ rare arrows each, then we have
$\pi_1 \equiv \pi_2$.
\end{clm}

Once we have established this, we can conclude the proof of
Claim~\ref{clm:hard}. To do so, take any two coterminal paths $\pi_1$ and $\pi_2$ that
are $n,p'$-equivalent.
By Claim~\ref{clm:choosenspec},
we can pick a threshold $n'$ (depending on~$u_1$, $u_2$, and $n$) such that $0 <
n' \leq n$ which is an 
$\card{S}$-distant rare-frequent threshold for~$\pi_1$ and~$\pi_2$.
By Claim~\ref{clm:multiple}, as $\pi_1$ and $\pi_2$ are $n,p'$-equivalent, we know that
they are also $n',p'$-equivalent.
Now that we have fixed the threshold~$n'$, following
Definition~\ref{def:rare}, 
we call an arrow of~$B$ \emph{rare} in~$u_1$ and $u_2$ if it occurs $\leq n'$
times in each,
and \emph{frequent} otherwise. Let $r_0$ be the number of rare arrows in $u_1$
and $u_2$: this number is the same for both, because they are $n',p'$-equivalent.
We now apply the outer inductive claim (Claim~\ref{clm:outer}) on $\pi_1$,
$\pi_2$,
$n'$, and $r = r_0$, to conclude that $\pi_1 \equiv \pi_2$, which concludes the
proof of Claim~\ref{clm:hard}.

It remains to show Claim~\ref{clm:outer}, which we do in the rest of the
section.

\subsection{Outer base case (B): all arrows in~\texorpdfstring{$\pi_1$}{π₁} and~\texorpdfstring{$\pi_2$}{π₂} are frequent.}
We now show Claim~\ref{clm:outer}.
The base case of the outer induction is for paths that contain no rare
arrows:

\begin{toappendix}
  \section{Self-contained proof of Lemma~\ref{lem:graph}}
  \label{apx:base}
\end{toappendix}

\begin{clm}
  \label{clm:base}
  For any two paths $\pi_1$ and $\pi_2$,
  for any $\card{S}$-distant rare-frequent threshold $n'$ for $\pi_1$ and $\pi_2$,
  if $\pi_1$ and $\pi_2$ are $n',p'$-equivalent and contain no rare arrows, then we have
  $\pi_1 \equiv \pi_2$.
\end{clm}

We prove this result by an induction over the number of different arrows that
occur in the paths $\pi_1$ and $\pi_2$, noting that thanks to
$n',p'$-equivalence an arrow occurs in $\pi_1$ iff it occurs in $\pi_2$. Specifically,
we show by induction on the integer $f$ the following:

\begin{clm}[Inner inductive claim on~$f$]
  \label{clm:inner}
  For any two paths $\pi_1$ and $\pi_2$,
  for any $\card{S}$-distant rare-frequent threshold $n'$ for $\pi_1$ and $\pi_2$,
  if $\pi_1$ and $\pi_2$ are $n',p'$-equivalent and contain no rare arrows and there
  are $\leq f$
different frequent arrows that occur, then we have $\pi_1 \equiv \pi_2$.
\end{clm}

Once we have established this, we can conclude the proof of the base case of the
outer induction, Claim~\ref{clm:base}. To do so, take any two coterminal paths
$\pi_1$ and $\pi_2$ and an $\card{S}$-distant rare-frequent threshold $n'$ such that
$\pi_1$ and $\pi_2$ are $n',p'$-equivalent, let $f_0$ be the number of frequent arrows
that occur in $\pi_1$ and in $\pi_2$, and conclude using Claim~\ref{clm:inner} with
$\pi_1$, $\pi_2$, $n'$, and $f = f_0$.

We now explain the proof of the inner inductive claim (Claim~\ref{clm:inner}), before moving on to the
inductive case of the outer induction.
The base case of Claim~\ref{clm:inner} with $f = 0$ is trivial as $\pi_1$ and
$\pi_2$ are then empty.
For the induction step,
assume that the claim holds for any $\pi_1$, $\pi_2$, and $n'$ such that there are
$\leq f$ different frequent arrows that occur.
Fix $\pi_1$, $\pi_2$, and $n'$ where there are $\leq f+1$ frequent arrows that
occur. Consider the multigraph $G$ of all arrows 
of~$S_\EE$
that occur in $\pi_1$ and $\pi_2$: $G$ has $\leq f+1$ edges.
If it has $\leq f$ edges we immediately conclude by induction hypothesis, so
assume it has exactly $f+1$ edges. 
Recall that $G$ is strongly connected:
indeed, as all arrows are frequent,
thanks to the existence of the $\card{S}$-distant rare-frequent threshold $n'$
(which is in particular 1-distant),
we know by Claim~\ref{clm:usccs}
that $G$ is a union of SCCs. Further, $\pi_1$
(or $\pi_2$) is a path where all edges of~$G$ occur, so it witnesses that $G$
is connected, and $G$ is therefore strongly connected.
The induction case is shown using a so-called \emph{ear decomposition} result on strongly connected
multigraphs:

\begin{lemmarep}[\cite{bang2008digraphs}]
\label{lem:graph}
  Let $G$ be a strongly connected nonempty directed multigraph. 
  We have:
  \begin{itemize}
    \item $G$ is a simple cycle; or
    \item $G$ contains a simple cycle $v_1 \rightarrow \cdots \rightarrow v_n
      \rightarrow v_1$
      with $n\geq 1$, where all vertices $v_1, \ldots, v_n$ are pairwise
      distinct, such that all intermediate vertices $v_2, \ldots, v_{n-1}$ only
      occur in the edges of the cycle, and such that the removal of the cycle
      leaves the graph strongly connected (note that the case $n=1$ corresponds
      to the removal of a self-loop); or
    \item $G$ contains a simple path $v_1 \rightarrow \cdots \rightarrow v_n$
      with $n \geq 2$ where
      all vertices are pairwise distinct, such
      that all intermediate vertices $v_2, \ldots, v_{n-1}$ only occur in the
      edges of the path, and such that the removal of the path leaves the graph
      strongly connected (note that the case $n=2$ corresponds to the removal of
      a single edge).
  \end{itemize}
\end{lemmarep}

\begin{toappendix}
  We repeat here that the result is standard. The proof given below is only for
  the reader's convenience, and follows~\cite{bang2008digraphs}.

\begin{proof}
  This result is showed using the notion of an ear decomposition of a directed
  multigraph. Specifically, following Theorem~5.3.2 of~\cite{bang2008digraphs},
  for any nonempty strongly connected multigraph $G$, we can build a copy of it
  (called~$G'$) by the following sequence of steps, with the invariant that~$G'$
  remains strongly connected:
  \begin{itemize}
    \item First, take some arbitrary simple cycle in~$G$ and copy it to~$G'$;
    \item Second, while there are some vertices of~$G$ that have not been copied
      to~$G'$, then pick some vertex $v$ of~$G$ that was not copied, such that
      there is an edge~$(v', v)$ in~$G$ with~$v'$ a vertex that was copied. Now
      take some shortest path (hence a simple path) $v \rightarrow \cdots
      \rightarrow v''$ from~$v$ to the subset of the vertices of~$G$ that had
      been copied to~$G'$. This path ends at a vertex~$v''$ which may or may not
      be equal to~$v'$. If $v'' \neq v'$, then we have a simple path $v'
      \rightarrow v \rightarrow \cdots \rightarrow v''$, which we copy to~$G'$;
      otherwise we have a simple cycle, which we copy to~$G'$. Note that, in
      both cases, all intermediate vertices in the simple path or simple cycle
      that we copy only occur in the edges of the path or cycle (as they had not
      been previously copied to~$G'$). Further, $G'$ clearly remains strongly
      connected after this addition.
    \item Third, once all vertices of~$G$ have been copied to~$G'$, take each
      edge of~$G$ that has not been copied to~$G'$ (including all self-loops),
      and copy it to~$G'$ (as a simple path of length~$1$). These additions
      preserve the strong connectedness of~$G'$.
  \end{itemize}
  At the end of this process, $G'$ is a copy of~$G$.

  Now, to show the result, take the graph~$G$, consider how we can
  construct it according to the above process, and distinguish three cases:
  \begin{itemize}
    \item If the process stopped at the end of the first step, then $G$ is a
      simple cycle (case 1 of the statement).
    \item If the process stopped after performing a copy in the second step,
      then considering the last simple path or simple cycle that we added, then
      it satisfies the conditions and its \emph{removal} from~$G$ gives a graph
      which is still strongly connected (case 2 or case 3 of the statement).
    \item If the process stopped after performing a copy in the third step, then
      considering the last edge that we added, then it is a simple path of
      length~$1$ and its \emph{removal} from~$G$ gives a graph which is still
      strongly connected (case 3 of the statement).
  \end{itemize}
  This concludes the proof.
\end{proof}
\end{toappendix}

This is a known result~\cite{bang2008digraphs}, but we give a self-contained
proof in Appendix~\ref{apx:base} for the reader's convenience.

Thanks to this result, we can distinguish three cases in the inner induction
step: if $G$ is non-empty, it must be a simple cycle, contain a removable simple cycle, or
contain a removable simple path.  We first give a high-level view of the argument in each
case:

\begin{proofsketch}
  The first case (B.1) is when
$G$ is a simple cycle. In this case $n',p'$-equivalence ensures that the
cycle is taken by $\pi_1$ and~$\pi_2$ some number of times with the same remainder
modulo~$p'$, so they evaluate to the same element because $p'$ is a multiple of
the period.

  The second case (B.2) is when $G$ contains a removable simple cycle. This time, we argue as in the previous case that the number of
occurrences of the cycle must have the same remainder, and we can use
Corollary~\ref{cor:zgcommutes} to merge all the occurrences together. However,
to eliminate them, we need to use Lemma~\ref{lem:prefixfrequent}, to modify~$\pi_1$
and~$\pi_2$ to have the same prefix (up to and including the cycle occurrences),
while preserving equivalence. This allows us to consider the rest of the paths
(which contains no occurrence of the cycle), apply the induction hypothesis to them, and conclude by compositionality.
A technicality is that we must ensure that removing the common prefix does not
make some arrows insufficiently frequent relative to the distant rare-frequent
threshold. We avoid this using Lemma~\ref{lem:spawnloop} to
insert sufficiently many copies of a suitable loop.

  The third case (B.3) is when $G$ contains a removable simple path $\tau$. The
reasoning is similar, but we also use Lemma~\ref{lem:spawnloop} to insert a loop
involving a return path for $\tau$ and a path that is parallel to~$\tau$ (i.e.,
does not share any arrows with it). The return path in this loop can then be
combined with~$\tau$ to form a loop, which we handle like in the previous case.
\end{proofsketch}

We now give the detailed argument for each case.

\paragraph*{Case B.1: $G$ is a simple cycle.}
If $G$ is a simple cycle, then distinguish the starting object of $\pi_1$
(hence, of $\pi_2$) as $o$, and
let $\alpha$ be the 
cycle from $o$ to
  itself, and $\tau$ the 
  path from $o$ to
  the common ending object of~$\pi_1$ and~$\pi_2$. We have: $\pi_1 = \alpha^{n_1}
  \tau$ and $\pi_2 = \alpha^{n_2} \tau$ with
$n_1$ and $n_2$ being $\geq n'-1$ and having the same remainder~$r$ modulo~$p'$.
We use the loop insertion lemma (Lemma~\ref{lem:spawnloop}) to insert
$\alpha^\omega$: the lemma tells us that
$\pi_1$ and $\pi_2$ respectively evaluate to the same category element as 
$\pi_1' = \alpha^{\omega+n_1}$  and 
$\pi_2' = \alpha^{\omega+n_2}$. Further, $\pi_1$ and $\pi_1'$, and $\pi_2$ and
$\pi_2'$, are $n',p'$-equivalent.

As $\alpha$ is a loop on the idempotent $o$, we know that
there exists an element $m\in oSo$ such that
$\alpha\equiv (o, m, o)$.
Hence, $\pi_1' = \alpha^{\omega+n_1} \tau \equiv (o, m^{\omega+n_1}, o) \tau$
(resp.\ $\pi_2'=
\alpha^{\omega+n_2} \tau \equiv (o, m^{\omega+n_2},
o) \tau$).
Now, we have
$m^{\omega+n_1} = m^{\omega+k_1\cdot p' + r}$ and
$m^{\omega+n_2} = m^{\omega+k_2\cdot p' + r}$ where $k_1$ and $k_2$ are the respective
quotients of $n_1$ and $n_2$ in the Euclidean division by~$p'$ and where 
 $r$ is the common remainder modulo~$p'$.
 By definition of the period~$p'$, we have $m^{\omega+n_1} = m^{\omega + r}$ and
$m^{\omega+n_2} = m^{\omega + r}$. We conclude that $m^{\omega+n_1} =
m^{\omega+n_2}$, so that
$\pi_1' \equiv \pi_2'$, and $\pi_1 \equiv \pi_2$. This
 concludes case B.1.

\paragraph*{Case B.2: $G$ has a simple cycle.}
Recall that, in this case, we know that $G$ has a simple cycle whose intermediate objects have no other incident edges and such that the removal of the simple cycle leaves the graph strongly connected.
Let $\alpha$ be the simple cycle, starting from the only object $o$ of the cycle having other incident edges.
We can then decompose $\pi_1$ and $\pi_2$ to isolate the occurrences of the simple
cycle (which must be taken in its entirety), i.e.:
\begin{align*}
		\pi_1 & = x_1 \alpha x_2 \alpha x_3 \cdots x_{t-1} \alpha x_t\\
		\pi_2 & = y_1 \alpha y_2 \alpha y_3 \cdots y_{t'-1} \alpha y_{t'}
\end{align*}
This ensures that the edges of~$\alpha$ do not occur elsewhere than in the $\alpha$
factors, except possibly in~$x_1$, $y_1$ and in $x_t$, $y_{t'}$ if the paths
$\pi_1$ and/or $\pi_2$ start and/or end in the simple cycle. This being said, in that case,
we know that the prefixes of~$\pi_1$ and $\pi_2$ containing this incomplete subset
of the cycle must be equal (same sequence of arrows), and likewise for their suffixes.
For this reason, it suffices to show the claim that $\pi_1$ and $\pi_2$ evaluate to
the same category element under the assumption that both their starting and
ending objects are not intermediate vertices of the cycle. The claim then
extends to the general case, by adding the common prefixes and suffixes
to the two paths that satisfy the condition, using compositionality of the
congruence. Thus, in the rest of the proof for this case, we
assume that the edges of $\alpha$ only occur in the~$\alpha$ factors.

We will now argue that, to show that $\pi_1 \equiv \pi_2$,
it suffices to show the same of two $n',p'$-equivalent coterminal paths from which all
occurrences of the edges of the cycle have been removed and where all other
edges still occur sufficiently many times. As this deals with paths where the underlying
multigraph contains fewer edges, the induction hypothesis will conclude.

To do this, by Lemma~\ref{lem:prefixfrequent}, as $x_1$ and $y_1$ are coterminal and
consist only of frequent arrows, and as the starting object of~$\alpha$ occurs
again in both paths, the path $\pi_1$ is $n',\omega$-equivalent,
and evaluates to the same category element as, some path:
\[\pi_1' = y_1 \alpha x_2' \alpha x_3' \cdots x_{t''-1}' \alpha x_{t''}'\]
The above is also $n',p'$-equivalent to~$\pi_1$, because $p'$ divides~$\omega$.
Thus, up to replacing $\pi_1$ by~$\pi_1'$, we can assume that $x_1 = y_1$.

Now, furthermore, $x_2,\ldots,x_{t-1}$ (resp.\ $y_2,\ldots,y_{t'-1}$) and $\alpha$ are coterminal
cycles over the object $o$ (which by definition corresponds to an idempotent of~$S$). Hence,
$\alpha x_2 \alpha x_3 \cdots x_{t-1} \alpha \equiv (o, m m_2 m m_3 \cdots m_{t-1} m, o)$
where $\alpha \equiv (o, m, o)$, $x_i \equiv (o, m_i, o)$ for $2 \leq i \leq t-1$ and
where $m$ and all $m_i$'s are in $oSo$, which is by hypothesis
a monoid in $\ZG$.
Now, 
as the arrows of $\alpha$ are frequent,
each one of them must occur $>n'$ times, so $\alpha$ (which contains exactly one
occurrence of each of these arrows) must occur $>n'$ times, and as $n'$ is a
$\card{S}$-distant rare-frequent threshold we clearly have by
Definition~\ref{def:distant} that $n' \geq \card{S}$, so $\alpha$ occurs $\geq
\card{S}+1$ times. Hence,
by Corollary~\ref{cor:zgcommutes}, we know that $m m_2 m m_3 \cdots m_{t-1}
m = m^{t-1} m_2 m_3 \cdots m_{t-1}$, because $\card{S}+1 \geq \card{oSo}+1$.
By applying the same reasoning to $\pi_2$, it suffices to show that the two following paths evaluate to the same category element, where $x_1 = y_1$:
\begin{align*}
		x_1 \alpha^{t-1} x_2 x_3 \cdots x_{t-1} x_t\\
		y_1 \alpha^{t'-1} y_2 y_3 \cdots y_{t'-1} y_{t'}
\end{align*}
Now, because these two paths are $n',p'$-equivalent, we know that $t-1$ and $t'-1$
have the same remainder modulo~$p'$. By the same reasoning as in case~B.1, they evaluate to the same category element as $\alpha^r$, where $r$ is the remainder.
So it suffices to show that the two following paths evaluate to the same category element, with $x_1 = y_1$:
\begin{align*}
		x_1 \alpha^r x_2 x_3 \cdots x_{t-1} x_t\\
		y_1 \alpha^r y_2 y_3 \cdots y_{t'-1} y_{t'}
\end{align*}
We now intend to use the induction hypothesis, but for this, we need to ensure
that $n'$ is still an $\card{S}$-distant rare-frequent threshold on the paths to
which we apply it. Specifically, we need to ensure that the edges not in
$\alpha$ still occur sufficiently many times.
To this end, let $\beta$ be any loop on~$o$ that visits all edges of~$G$ except the ones
in~$\alpha$: this is doable because $G$ is still strongly connected after the
removal of~$\alpha$. Up to exponentiating $\beta$ to some power
$\beta^{n''\omega}$, we can assume that $\beta$
traverses each edge sufficiently many times to satisfy the lower bound imposed
by the requirement of~$n'$ being an $\card{S}$-distant rare-frequent threshold.
By Lemma~\ref{lem:spawnloop}, it suffices to show that the
following paths evaluate to the same category element:
\begin{align*}
  \pi_1' & = x_1 \alpha^r \beta^{n''\omega} x_2 x_3 \cdots x_{t-1} x_t\\
  \pi_2' & = y_1 \alpha^r \beta^{n''\omega} y_2 y_3 \cdots y_{t'-1} y_{t'}
\end{align*}
Note that $n',p'$-equivalence is preserved because $n''$ is a multiple
of~$\omega$, hence of~$p'$.
Now, observe that both paths start by $x_1 \alpha = y_1 \alpha$, and the arrows
of~$\alpha$ do not occur in the rest of the paths. Consider the paths
$\beta^{n''\omega}x_2 x_3 \ldots x_t$ and $\beta^{n''\omega}y_2 y_3 \ldots y_{t'}$.
 They are paths that are coterminal, $n',p'$-equivalent because $\pi_1$ and $\pi_2$
 were, where the frequent letters that are used are a strict subset of the ones used
 in~$\pi_1$ and~$\pi_2$, and where all other frequent letters occur sufficiently many
 times for $n'$ to still be an $\card{S}$-distant rare-frequent threshold (as
 guaranteed by $\beta^{n''\omega}$). Thus,
 by induction hypothesis of the inner induction, we know that these two paths evaluate to the same
 category element, so that $\pi_1'$ and $\pi_2'$ also do. This concludes case B.2.

\paragraph*{Case B.3: $G$ has a simple path.}
Recall that, in this case, we know that $G$ has a simple path where the starting
and ending objects of intermediate arrows have no other incident edges, and such that the removal of the simple path leaves the graph strongly connected.
We denote the path by $\tau$ and denote by $x \neq y$ its starting and ending objects.
Since the removal of the path does not affect strong connectedness of the
graph, there is a simple path from $x$ to $y$ sharing no edges with $\tau$, which
we denote by~$\kappa$.
Furthermore, there is a simple path from~$y$ to~$x$ sharing no edges
with~$\tau$ (this is because all intermediate objects of~$\tau$ only occur in
the edges of~$\tau$), which we denote by~$\rho$.

Like in the previous case, up to removing common prefixes and suffixes, it
suffices to consider the case where $\pi_1$ and $\pi_2$ do not start or end in the
intermediate vertices of~$\tau$.
For that reason, isolating all occurrences of~$\tau$ also isolates all
occurrences of the edges of~$\tau$, and we can write:
\begin{align*}
        \pi_1 & = x_1 \tau x_2 \tau x_3 \cdots x_{t-1} \tau x_t\\
        \pi_2 & = y_1 \tau y_2 \tau y_3 \cdots y_{t'-1} \tau y_{t'}
\end{align*}
where the $x_i$ and $y_i$ do not use the edges of~$\tau$.
Like in the previous case, by Lemma~\ref{lem:prefixfrequent}, we can assume that $x_1 = y_1$.

By Lemma~\ref{lem:spawnloop},
we insert a loop $(\rho\kappa)^{\omega}$ after every
occurrence of $\tau$
without changing the category element and still respecting the
$n',\omega$-congruence,
hence the $n',p'$-congruence because $p'$ divides~$\omega$.
By expanding $(\rho\kappa)^{\omega}=\rho\kappa(\rho\kappa)^{\omega-1}$, it suffices to show that the following paths evaluate to the same category element, with $x_1 = y_1$:
\begin{align*}
        \pi_1' & = x_1 \tau \rho \kappa (\rho\kappa)^{\omega-1} \cdots  x_{t-1} \tau \rho
        \kappa (\rho\kappa)^{\omega-1}  x_t\\
        \pi_2' & = y_1 \tau \rho \kappa (\rho\kappa)^{\omega-1} \cdots y_{t'-1} \tau \rho
        \kappa (\rho\kappa)^{\omega-1}  y_{t'}
\end{align*}

We can now regroup the occurrences of $\tau \rho$, which are loops such that
some edges (namely, the edges of~$\tau$) only occur in these factors.
This means that we can conclude as in case 2 for the cycle $\tau \rho$, as this
cycle contains some edges that only occur there; we can choose~$\beta$ at the
end of the proof to be a loop on~$x$ visiting all edges of~$G$ except those
of~$\tau$,
which is again possible because $G$ is still strongly connected even after the
removal of~$\tau$.

This establishes case 3 and concludes the induction step of the proof,
establishing Claim~\ref{clm:inner}.

We have thus proved by induction that $\pi_1$ and $\pi_2$ evaluate to the same
category element, in the base case (Claim~\ref{clm:base}) of the outer induction
(Claim~\ref{clm:outer}) where all edges of~$\pi_1$ and~$\pi_2$ are frequent.

\subsection{Outer induction step (I): some arrows are rare.}

Let us now show the induction step for the outer induction
(Claim~\ref{clm:outer}), namely, the
induction on the number of occurrences of rare arrows.
We assume the claim of the outer induction for $r \in \NN$.
Consider two paths $\pi_1$ and $\pi_2$ and an $\card{S}$-distant rare-frequent
threshold $n'$ such that $\pi_1$ and $\pi_2$ are $n',p'$-equivalent and such that they contain
$r+1$ rare letters. Let us partition them as $\pi_1 = q_1 a s_1$
and
$\pi_2 = q_2 a s_2$
where $q_1$ and $q_2$ all consist of frequent arrows, and $a$ is the first rare
arrow of $\pi_1$ and $\pi_2$ (note that $n',p'$-equivalence implies that the first rare
arrow is the same in both paths).
In this case, $q_1$ and $q_2$ are two coterminal paths consisting only of
frequent arrows (or they are empty), and $s_1$ and $s_2$ are two coterminal
paths (possibly empty) with $r$ rare letter occurrences.

Remember that, as $n'$ is an $\card{S}$-distant rare-frequent threshold for
$\pi_1$
and $\pi_2$, then we know that the frequent arrows of $\pi_1$ form a union of SCCs
(Claim~\ref{clm:usccs}); note that, thanks to
$n',p'$-equivalence, the same is true of~$\pi_2$ with the same SCCs.
Consider the SCC $C$ of frequent arrows that contains the starting object of~$a$. There are two cases, depending on whether some object of~$C$ occurs again in~$s_1$ or not.
Note that some object of~$C$ occurs again in~$s_1$ iff the same is true
of~$s_2$, because which frequent arrow components occur again is entirely
determined by the ending objects of the rare arrows of~$s_1$ and $s_2$, which
are identical thanks to $n',p'$-equivalence.

\paragraph*{Case I.1: $C$ occurs again after~$a$.}
In this case, we are in a situation where we can apply Lemma~\ref{lem:prefixfrequent}, because $q_1$ and $q_2$ only
consist of frequent arrows and some object of the SCC~$C$ of the starting object
of~$a$ occurs again in~$s_1$. The lemma tells us that
there is a path:
\[
\pi_1' = q_2 a s_1'
\]
which evaluates to the same category element as~$\pi_1$ and is $n',\omega$-equivalent to
it, hence $n',p'$-equivalent.
Hence, by compositionality, it suffices to show that $s_1'$ and $s_2$ evaluate to the same category element.

To apply the induction hypothesis, we simply need to ensure that $n'$ is still a
$\card{S}$-distant rare-frequent threshold for~$s_1'$ and~$s_2$. To do this, we
need to ensure that the arrows that are frequent in~$\pi_1$ and~$\pi_2$ are still
frequent there, and still satisfy the $\card{S}$-distant condition. Fortunately,
we can simply ensure this by inserting a loop using Lemma~\ref{lem:spawnloop}.
Formally, write $s_1' = r_1 t_1$ where the intermediate object is the object of the
SCC~$C$ that occurred in~$s_1$ (the existence of such a decomposition is
a consequence of the statement of Lemma~\ref{lem:prefixfrequent}), and write $s_2 = r_2
t_2$ in the same way (which we already discussed must be possible with~$s_2$).
Let $\beta$ be an arbitrary loop of frequent arrows where all arrows of~$C$ occur:
this is possible because $C$ is strongly connected. We know by
Lemma~\ref{lem:spawnloop} that 
$s_1' = r_1 t_1$ and $r_1 \beta^{\omega} t_1$ are both
$n',\omega$-equivalent, hence $n',p'$-equivalent, and evaluate to the same
category element: this is also true with
$w_1 \colonequals r_1 \beta^{n'' \omega} t_1$ for a sufficiently large~$n''$ such that every frequent arrow
of~$C$ occurs as many times as it did in~$\pi_1$. Likewise, $s_2 = r_2 t_2$ and
$w_2 \colonequals r_2 \beta^{n'' \omega} t_2$ are both $n',p'$-equivalent and evaluate to the same category
element. So it suffices to consider $w_1$ and $w_2$.

  Let us apply the induction hypothesis to them. They are two coterminal paths,
and they are $n',p'$-equivalent because $w_1 \sim_{n',p'} \pi_1' \sim_{n',p'} \pi_1$ and $w_2
\sim_{n',p'} \pi_2$ and by hypothesis $\pi_1 \sim_{n',p'} \pi_2$. What is more, the arrows
that were rare in~$\pi_1$ and~$\pi_2$ are still rare for them, and they have $r$
occurrences in total: this was true by construction of $s_1$ and $s_2$ and is
true of~$s_1'$ because $\pi_1 = q_1 a s_1 \sim_{n',p'} q_2 a s_1'$ and all arrows
of~$q_2$ are frequent so the rare subwords of~$s_1$ and~$s_1'$ are the same.
The arrows that were frequent in~$\pi_1$ and~$\pi_2$ are still frequent in $s_1$ and
$s_2$ and occur at least as many times as they did in~$\pi_1$ and $\pi_2$
respectively: we have guaranteed this for the arrows of~$C$ by inserting
$\beta^{n''\omega}$, and this is clear for the arrows outside of~$C$ as
all their occurrences in~$\pi_1$ and~$\pi_2$ were in $s_1$ and $s_2$ respectively,
and $s_1'$ has at least as many occurrences of every letter as $s_1$ does (this
is a consequence of the statement of Lemma~\ref{lem:prefixfrequent}). This ensures
that $s_1' \sim_{n',p'} s_2$, and that $n'$ is still an $\card{S}$-distant
rare-frequent threshold for them.

Hence, by the induction hypothesis, we have $s_1' \equiv s_2$, so that by
compositionality we have $\pi_1 \equiv \pi_2$.

\paragraph*{Case I.2: $C$ does not occur again after~$a$.}
In this situation, we cannot apply Lemma~\ref{lem:prefixfrequent}. However,
intuitively, the arrows visited in $q_1$ and in $q_2$ must be disjoint from
those visited in the rest of the paths, so we can independently reason on $q_1$
and $q_2$, and on $s_1$ and $s_2$.

Formally, we first claim that $q_1 \equiv q_2$ by the base case (Claim~\ref{clm:base}) of the
outer induction (Claim~\ref{clm:outer}). Indeed, first note that they are two coterminal paths. Now,
there are two cases: every arrow $x$ which is frequent in $\pi_1$ and $\pi_2$ is
either in the SCC~$C$ of the
starting object of~$a$ or not. In the first case, all the occurrences of~$x$
in $\pi_1$ must be in~$q_1$, as any occurrence of~$x$ in $s_1$ would witness that
we are in Case I.1; and likewise all its occurrences in~$\pi_2$ must be in~$q_2$.
In the second case, all its occurrences in $\pi_1$ must be in $s_1$ and all
its occurrences in~$\pi_2$ must be in~$s_2$, for the same reason. Thus, $q_1$ and
$q_2$ contain no letter which was rare in~$\pi_1$ and~$\pi_2$, some of the frequent
letters of~$\pi_1$ and~$\pi_2$ (those of the other SCCs) do not occur there at all,
and the others occur there with the same number of occurrences. Thus indeed $q_1
\sim_{n',p'} q_2$, they contain no rare arrows, and $n'$ is still
an~$\card{S}$-distant rare-frequent threshold for them. Thus, the base case of
the outer induction concludes that they
evaluate to the same category element.

We now claim that $s_1 \equiv s_2$ by the induction case of the outer induction
(Claim~\ref{clm:outer}).
Indeed, they are again two coterminal paths. What is more, by the previous
reasoning, the arrows that are frequent in~$\pi_1$ and~$\pi_2$ either occur only
in~$s_1$ and~$s_2$ or do not occur there at all. Thus, $s_1$ and $s_2$ contain
$r$ rare arrows (for the arrows that were already rare in~$\pi_1$ and~$\pi_2$),
and the frequent arrows either occur in~$s_1$ and $s_2$ with the same number of
occurrences as in~$\pi_1$ and~$\pi_2$ or not at all. This implies that $n'$ is still
an $\card{S}$-distant rare-frequent threshold for~$s_1$ and~$s_2$. Thus, we have
$s_1 \sim_{n',p'} s_2$ and the induction case of the outer induction establishes
that $s_1 \equiv s_2$.

Thus, by compositionality, we know that $\pi_1$ and $\pi_2$ evaluate to the same category element.
We have concluded both cases of the outer induction proof and shown
Claim~\ref{clm:outer}.

\subsection{Concluding the proof.}
We have proven Claim~\ref{clm:outer}, and explained afterwards how to use it to
show that $\pi_1$ and $\pi_2$ evaluate to the same category element.
This implies that $n',p'$-equivalence for our choice of~$n'$, hence also
$n,p$-equivalence, is compatible with~$S_\EE$. Thus, by
Corollary~\ref{cor:straubing} we know that $L$ is in $\ZGD_p$. Hence, $L \in
\LZG_p$
implies that $L \in \ZGD_p$, so we have shown Claim~\ref{clm:hard}. Together with
Claim~\ref{clm:easy}, it establishes the locality result $\LZG_p = \ZGD_p$,
and we have shown Theorem~\ref{thm:locality2}.

\section{Conclusion}
  \nosectionappendix
\label{sec:conclusion}
In this paper, we have given a characterization of the languages of~$\ZG$, and
proved that the variety $\ZG$ is local. More specifically, we have shown this
for all the varieties $\ZG_p$ for $p>0$, in particular $\MNil = \ZG_1 =\ZG\cap
\mathbf{A}$.

A natural question for further study is whether the variety~$\ZE$ is also local.
This question seems more complicated. Indeed, as proved by
Almeida~\cite{almeida1994finite}, we have $\ZE=\mathbf{G}\lor\Com$,
that is, $\ZE$ is the variety of monoids generated by the variety $\mathbf{G}$
of groups and the variety $\mathbf{Com}$ of commutative languages.
Now, $\mathbf{G}$ is a local variety~\cite[Example~1.3]{therien85},
while $\mathbf{Com}$ is not \cite[Example~1.4]{therien85}.
Further, as we have shown that $\LZG = \ZGD$ and $\Com$ is a subset of $\ZG$,
the counter-example languages (e.g., $e^*af^*be^*cf^*$) to the locality of
$\Com$ (i.e., that are in $\mathbf{LCom}$ but not in $\Com*\D$) cannot be
counter-examples to the locality of~$\ZE$ (because they are in $\mathbf{LCom}$,
hence $\LZG$, hence $\ZGD$, hence $\ZE*\D$).
This being said, if $\ZE$ is indeed local, a proof would probably require
different techniques from ours, given that we do not see how our
techniques could be used even to reprove the locality of~$\mathbf{G}$.

We hope that extending our approach to a study of locality for centrally defined
varieties in general could lead to such general results on the interplay of join
operations and of the locality or non-locality for arbitrary varieties, in the
spirit of the results shown in~\cite{costa2013some} for various Mal'cev
products.

\begin{toappendix}
\makeatletter
    \vspace{-2cm}%
    \nobreak%
    \insert\copyins{\hsize.57\textwidth
\vbox to 0pt{\vskip12 pt%
      \fontsize{6}{7\p@}\normalfont\upshape
      \everypar{}%
      \noindent\fontencoding{T1}%
  \headertextsf{This work is licensed under the Creative Commons
  Attribution License. To view a copy of this license, visit
  \texttt{https://creativecommons.org/licenses/by/4.0/} or send a
  letter to Creative Commons, 171 Second St, Suite 300, San Francisco,
    CA 94105, USA, or Eisenacher Strasse 2, 10777 Berlin, Germany}\vss}
      \par
      \kern\z@}%
\makeatother
\end{toappendix}

\bibliographystyle{alphaurl}
\bibliography{bib}
\end{document}